\documentclass[sigconf]{aamas}

\usepackage{balance} 
\usepackage{amsmath}
\usepackage{mathtools}
\usepackage{amsthm}
\usepackage{subcaption}
\usepackage{algorithm}
\usepackage{algpseudocode}

\usepackage[normalem]{ulem}

\setcopyright{ifaamas}
\acmConference[AAMAS '24]{Proc.\@ of the 23rd International Conference
on Autonomous Agents and Multiagent Systems (AAMAS 2024)}{May 6 -- 10, 2024}
{Auckland, New Zealand}{N.~Alechina, V.~Dignum, M.~Dastani, J.S.~Sichman (eds.)}
\copyrightyear{2024}
\acmYear{2024}
\acmDOI{}
\acmPrice{}
\acmISBN{}

\acmSubmissionID{<<EasyChair submission id>>}

\theoremstyle{plain}
\newtheorem{theorem}{Theorem}[section]
\newtheorem{proposition}[theorem]{Proposition}
\newtheorem{lemma}[theorem]{Lemma}
\newtheorem{corollary}[theorem]{Corollary}
\theoremstyle{definition}
\newtheorem{definition}[theorem]{Definition}

\theoremstyle{remark}
\newtheorem{remark}[theorem]{Remark}

\newcommand{\game}{\mathcal{G}}
\newcommand{\agentset}{\mathcal{N}}
\newcommand{\edgeset}{\mathcal{E}}
\newcommand{\actionset}[1]{\mathcal{S}_{#1}}

\newcommand{\NE}{\mathbf{\bar{x}}}

\newcommand{\A}{\mathbf{A}}

\newcommand{\defeq}{\vcentcolon=}

\newcommand{\x}{\mathbf{x}}
\newcommand{\e}{\mathbf{e}}
\newcommand{\y}{\mathbf{y}}

\newcommand{\zeros}{\mathbf{0}}

\newcommand{\U}{\mathbf{u}} 

\newcommand{\epsX}{\epsilon_X}
\newcommand{\epsY}{\epsilon_Y}

\newcommand{\R}{\mathbb{R}}

\newcommand{\X}{\mathcal{X}}

\acmSubmissionID{176}

\title[Multi-Agent Learning in Network Games]{On the Stability of Learning\\ in Network Games with Many Players}

\author{Aamal Hussain}
\affiliation{
  \institution{Imperial College London}
  \city{}
  \country{}}
\email{aamal.hussain15@imperial.ac.uk}

\author{Dan Leonte}
\affiliation{
  \institution{Imperial College London}
  \city{}
  \country{}}
\email{dan.leonte16@imperial.ac.uk}

\author{Francesco Belardinelli}
\affiliation{
  \institution{Imperial College London}
  \city{}
  \country{}}
\email{francesco.belardinelli@imperial.ac.uk}

\author{Georgios Piliouras}
\affiliation{
  \institution{Singapore University of Technology and Design}
  \city{}
  \country{}}
\email{georgios@sutd.edu.sg}

\begin{abstract}
Multi-agent learning algorithms have been shown to display complex, unstable behaviours in a wide array of games. In fact, previous works indicate that convergent behaviours are less likely to occur as the total number of agents increases. This seemingly prohibits convergence to stable strategies, such as Nash Equilibria, in games with many players.

To make progress towards addressing this challenge we study the Q-Learning Dynamics, a classical model for exploration and exploitation in multi-agent learning. In particular, we study the behaviour of Q-Learning 
on games where interactions between agents are constrained by a network. We determine a number of sufficient conditions, depending on the game and network structure, which guarantee that agent strategies converge to a unique stable strategy, called the Quantal Response Equilibrium (QRE). Crucially, these sufficient conditions are independent of the total number of agents, allowing for provable convergence in arbitrarily large games. 

Next, we compare the learned QRE to the underlying NE of the game, by showing that any QRE is an $\epsilon$-approximate Nash Equilibrium. We first provide tight bounds on $\epsilon$ and show how these bounds lead naturally to a centralised scheme for choosing exploration rates, which enables independent learners to learn stable approximate Nash Equilibrium strategies. We validate the method through experiments and demonstrate its effectiveness even in the presence of numerous agents and actions. Through these results, we show that independent learning dynamics may converge to approximate Nash Equilibria, even in the presence of many agents.
\end{abstract}

\keywords{Multi-Agent Learning, Quantal Response Equilibrium, Online Learning in Games}

\newcommand{\BibTeX}{\rm B\kern-.05em{\sc i\kern-.025em b}\kern-.08em\TeX}

\begin{document}


\pagestyle{fancy}
\fancyhead{}


\maketitle 


\section{Introduction}

Game Theory (EGT) has emerged as a powerful formalism for studying learning in multi-agent settings \cite{schwartz:MARL,tuyls:foundational-models}. Here, agents are required to explore their state space to determine optimal actions, whilst simultaneously maximising their expected reward in the face of the changing behaviour of their opponents. By modelling these situations as idealised games it is possible to study the effect of various factors, such as payoffs and number of agents, on the dynamics of learning. An important question which is often studied from this lens is whether popular multi-agent learning algorithms converge to an equilibrium \cite{metrick:fp,harris:fp,piliouras:potential} (most often the Nash Equilibrium). 

Unfortunately, it seems that the general answer to this question is \emph{no}. Recent work has shown that, even in zero-sum games, the dynamics of no-regret learning algorithms can be cyclic \cite{piliouras:cycles} or chaotic \cite{cheung:vortices}. In addition, even small deviations from the zero-sum setting can result in robustly non-convergent dynamics \cite{cheung:decomposition,hussain:ijcai} so that in general-sum games, non-convergent behaviour appears to be the norm \cite{pangallo:bestreply,pangallo:taxonomy,galla:complex,flokas:donotmix,kleinberg:nashbarrier,imhof:cycles,galla:cycles,svs:chaos,sato:diversity,griffin:evonetworks}. To make matters worse, recent findings in \cite{sanders:chaos} suggest that, as the number of agents in the game increases, the likelihood for chaotic dynamics also increases when agents have low exploration rates. Similarly, the results of \cite{hussain:aamas} imply that incredibly large exploration rates may be required in games with many agents in order to ensure convergence. This seemingly presents a bottleneck for strong convergence guarantees in multi-agent settings with many agents. 

Despite this, many real world problems such as resource allocation \cite{amelina:load-balancing,grammatico:NAG}, routing \cite{Bielawski2021,Chotibut2019TheOne,Chotibut2019TheOptimistic} and robotics \cite{hamann:swarm,Shokri2020Leader-FollowerActiveness} consider a large number of agents who continuously adapt to one another. These practical applications in conjunction with the negative results in the face of many players immediately yield the following question: 
\begin{center}
\emph{Is there any hope for independent learning agents to converge to an equilibrium in games with many players?}
\end{center}

 To make progress in answering this question, this work examines multi-agent learning in \emph{network} games. Here, it is assumed that agents can only interact with their neighbours within an underlying communication network. Such systems are ubiquitous: machine learning architectures often impose structure between models \cite{hoang:mgan,li:triple-gan}; in robotic systems, agents interact through communication networks \cite{Grammatico2016DecentralizedControl,Shokri2020Leader-FollowerActiveness}; in both economics and biology, agent interactions are constrained through social networks. Network games refine the setting of \cite{sanders:chaos,hussain:aamas}, in which it was assumed that each agent is directly influenced by every other agent in the environment. This work provides strong evidence that the network structure matters, in some cases even more so than the total number of agents. 

\paragraph{Model and Contribution} We consider agents who update via the \emph{Q-Learning} dynamic, \cite{sato:qlearning,tuyls:qlearning}, a foundational model from game theory which describes the behaviour of agents who balance exploration and exploitation. Similar to \cite{hussain:aamas} we determine a number of sufficient conditions on exploration rates such that Q-Learning is guaranteed to converge to a unique equilibrium. In this work, however, we find that these conditions depend on graph theoretic properties of the interaction network. In our experiments, we examine how these conditions depend on the total number of
agents and find network structures for which there is no explicit dependence. These implications are visualised on a number of
representative network games and it is shown that large numbers of agents may converge to an equilibrium, so long as weakly connected network structures are used. By contrast, if the network is strongly connected, we recover the results of \cite{hussain:aamas,sanders:chaos} and show that stability depends on the total number of agents. 

The equilibrium solution to which Q-Learning converges is the \emph{Quantal Response Equilibrium} (QRE) \cite{mckelvey:qre,piliouras:zerosum}, a widely studied extension of the Nash Equilibrium for agents who explore their state space \cite{piliouras:potential,galstyan:2x2,gemp:sample}. In this work, we quantify the `distance' between a QRE and NE by showing that any QRE is an approximate Nash Equilibrium and providing tight bounds on this approximation. Using this, we present a procedure for choosing exploration rates so that Q-Learning agents may converge `closer' to the Nash Equilibrium, whilst maintaining the stability of the dynamic. We validate this procedure in a number of large scale network games and show that it leads to improvements in the convergence of Q-Learning dynamics towards approximate Nash Equilibria.

\paragraph{Related Work}
In \cite{galla:complex} the authors showed that the Experience Weighted Attraction (EWA) dynamic, which is closely related to Q-Learning \citep{piliouras:zerosum}, achieves chaos in classes of two-player games. Advancing this result, \cite{sanders:chaos} showed that chaotic dynamics become more prevalent as the number of agents
increase. Similar to this work,
\cite{hussain:aamas} apply the framework of \emph{monotone game} \cite{parise:network,facchinei:VI,tatarenko:monotone} to show that Q-Learning Dynamics converge to a unique equilibrium in any game, given sufficient exploration. However, they also find that this condition increases with the number of agents.

Besides online learning, other approaches have been developed to try to \emph{compute} Nash Equilibria in games. For our purposes, the most relevant of these are homotopy-like methods \cite{turocy:homotopy,herings:homotopy}. The principle of these methods is to perturb the payoff functions so that the resulting perturbed game is `easier' to solve. Then, by iteratively annealing this perturbation, one can approximate the underlying NE. Recently \cite{gemp:sample} applies an entropy perturbation of payoffs and use gradient-descent based approach to solve for a continuum of \emph{Quantal Response Equilibria} (QRE), which eventually leads to a NE \cite{mckelvey:qre}. Whilst homotopy methods present a powerful tool for computing approximate equilibria, they often lack the advantages of decentralisation provided by online learning and may not come with strong guarantees. \cite{Perolat2020FromRegularization} combines the entropy perturbation approach with online learning and show that, in two-player zero-sum games, this method allows independent learners to converge asymptotically to an NE. However, as with most learning strategies, its behaviour in many player, general sum games is unknown.

We address the problem of learning in many player games by examining the role of an underlying communication network. A number of works in game theory have shown that network structure affects the uniqueness and stability of NE \cite{ballester:whoiswho,bramoulle:networks,parise:network,melo:network,piliouras:poincare}. Our main result refines that of \cite{hussain:aamas} to include the network and show that Q-Learning dynamics can reach a QRE in any network game, given sufficiently high exploration rates. Crucially, these conditions are explicitly independent of the total number of agents. We also show that the QRE achieved by Q-Learning is an approximate Nash Equilibrium, and design a centralised scheme for updating exploration rates so that Q-Learning dynamics converge along the continuum of stable QRE to an approximate Nash Equilibrium.

\section{Preliminaries} \label{sec::Prelims}
    
We begin in Section \ref{sec::model} by defining the network game model, which is the setting on
which we study the Q-Learning dynamics, which we describe in Section \ref{sec::LearningModel}.

\subsection{Game Model}\label{sec::model}

In this work, we consider \textit{network polymatrix games} \cite{piliouras:zerosum}. A Network Game
is described by the tuple $\game = (\agentset, \edgeset, (u_k, \actionset{k})_{k \in \agentset})$,
where $\agentset$ denotes a finite set $\agentset$ of players, indexed by $k = 1, \ldots, N$. Each
agent can choose from a finite set $\actionset{k}$ of actions, indexed by $i = 1, \ldots, n$. We
denote the \emph {strategy} $\x_k$ of an agent $k$ as the probabilities with which they play their
actions. Then, the set of all strategies of agent $k$ is $\Delta(\actionset{k}) := \left\{ \x_k \in
\R^{n} \, : \, \sum_i x_{ki} = 1, x_ {ki} \geq 0 \right\}$. Each agent is also given a payoff
function $u_k \, : \Delta(\actionset{k}) \times \Delta(\actionset{-k}) \rightarrow \R$.
Agents are connected via an underlying network defined by $\edgeset$. In particular, $\edgeset$
consists of pairs $(k, l) \in \agentset \times \agentset$ of connected agents
$k$ and $l$. For any agent $k \in \agentset$, we denote by $\agentset_k = \{ l
\in \agentset : (k, l) \in \edgeset\}$ the \emph{neighbours} of $k$, i.e.~all
the agents who directly interact with agent $k$ in the network. An
equivalent way to define the network is through an \emph{adjacency matrix} $G$ such that
\begin{equation*}
    [G]_{k,l} = \begin{cases}
        1, \text{ if agents $k, l$ are connected,} \\
        0, \text{ otherwise.}
    \end{cases}.
\end{equation*}
It is assumed that the network is undirected so that $G$ is a symmetric matrix. Each edge $(k, l)
\in \edgeset$ corresponds to a pair of payoff matrices $A^{kl}$, $A^{lk}$. With these
specifications, the payoff received by each agent $k$ under joint strategy $\x = (\x_k, \x_{-k})$ is given by
\begin{equation} \label{eqn::GPGPayoffs} u_k(\x_k, \x_{-k}) = \sum_{(k, l) \in \edgeset} \x_k \cdot A^{kl} \x_l.
\end{equation}
 For any $\x \in \Delta =: \times_k \Delta(\actionset{k})$, we can define the reward to agent $k$
 for playing action $i$ as $r_{ki}(\x_{-k}) = \partial u_{ki}(\x)/\partial x_{ki}$. Under
 this notation, $u_k(\x_k, \x_{-k}) = \langle \x_k, r_k(\x_{-k}) \rangle$. With this in place, we can
 define suitable equilibrium solutions for the game.

\begin{definition}[Nash Equilibrium (NE)]
    A joint mixed strategy $\NE \in \Delta$ is a
\emph {Nash Equilibrium} (NE) if, for all agents $k$ and all actions $i \in
\actionset{k}$
    \begin{equation*}
        \NE_{k} = \arg\max_{\y_k \in \Delta_k} u_k(\y_k, \NE_{-k}).
    \end{equation*}
\end{definition}

\begin{definition}[Quantal Response Equilibrium (QRE)] A joint mixed strategy $\NE \in \Delta$ is a
\emph {Quantal Response Equilibrium} (QRE) if, for all agents $k$ and all actions $i \in
\actionset{k}$
    \begin{equation*}
        \NE_{ki} = \frac{\exp(r_{ki}(\NE_{-k})/T_k)}{\sum_{j \in \actionset{k}} \exp(r_{kj}(\NE_{-k})/T_k)}.
    \end{equation*}
\end{definition}

The QRE \cite{camerer:bgt,mckelvey:qre} naturally extends the Nash Equilibrium through the
parameter $T_k$, known as the \emph{exploration rate}. In particular,
the limit $T_k \rightarrow 0$ corresponds exactly to the Nash Equilibrium, whereas the limit $T_k
\rightarrow \infty$ corresponds to the case where action $i \in \actionset{k}$ is
played with the same probability regardless of its associated reward. The link between the QRE and
the Nash Equilibrium is made precise through the following result.

\begin{proposition}[\cite{melo:qre}] \label{prop::melo-qre}
    Consider a game $\game = (\agentset, \edgeset, (u_k, \actionset{k})_{k \in \agentset})$ and let
$T_1, \ldots, T_N > 0$ be exploration rates. Define the perturbed game $\game^H = (\agentset,
\edgeset, (u_k^H, \actionset{k})_{k \in \agentset})$ with the payoff functions
\begin{equation*}
    u_k^H(\x_k, \x_{-k}) = u_k(\x_k, \x_{-k}) - T_k \langle \x_k, \ln \x_{k} \rangle.
\end{equation*}
Then $\NE \in \Delta$ is a QRE of $\game$ iff it is a Nash Equilibrium of $\game^H$.
\end{proposition}
    
\paragraph{Game Structure} To achieve our main result, we must parameterise
interactions in the network game. This allows us to consider network games which are not
necessarily zero-sum. First, we define the \emph{influence bound} for each agent $k$.
\begin{definition}[Influence Bound]
    Let $\game = (\agentset, \edgeset, (u_k, \actionset{k})_{k \in \agentset})$
    be a network game. Then, for any $k \in \agentset$, the \emph{influence
    bound} is given by
    \begin{equation}
        \delta_k = \max_{i \in \actionset{k}, a_{-k}, \Tilde{a}_{-k} \in \actionset{-k}} \{ |r_{ki}(a_{-k}) - r_{ki}
        ( \Tilde{a}_{-k})| \},
    \end{equation}
    where the pure strategies $a_{-k}, \Tilde{a}_{-k} \in \actionset{-k}$ differ only in
    the action of one agent $l \in \agentset_k$.
\end{definition}
The influence bound describes how sensitive each agent's reward is to changes in opponent strategies.
As another parameterisation which is directly applicable to network games, we define the \emph{intensity of identical interests}.
\begin{definition}[Intensity of Identical Interests]
    Let $\game$ be a network game whose edgeset $\edgeset$ is associated with the payoff matrices $(A^{kl}, A^{lk})_{(k, l) \in \edgeset}$.
    The \emph{intensity of identical interests} $\sigma_I$ of $\game$ is given as
    \begin{equation} \label{eqn::intensity}
        \sigma_I = \max_{(k, l) \in \edgeset} \lVert A^{kl} + (A^{lk})^\top \rVert_2, 
    \end{equation}
    where $\lVert M \rVert_2 = \sup_{\Vert\x\rVert_2 = 1} \lVert M\x \rVert_2$ denotes the operator
    $2$-norm \cite{meiss:book}.
\end{definition}
The intensity of identical interests can be thought of as a measure of how \emph{cooperative} a network
game is. The reasoning for this is as follows. Suppose $A, B$ are the payoff matrices which maximise (\ref{eqn::intensity}) and suppose that $B^\top = cA$ for some $c = (-1, 1)$. Then, $\sigma_I$ is minimised when $c = -1$, in which case $A, B$ is zero-sum, and is maximised at $c=1$ so that $A = B^\top$, which defines an game of identical interests.

\subsection{Learning Model} \label{sec::LearningModel}

In this work, we analyse the \emph{Q-Learning dynamic}, a prototypical model for determining
optimal policies by balancing exploration and exploitation \cite{sutton:barto,schwartz:MARL}. In this model, each agent $k \in
\agentset$ maintains a history of the past performance of each of their actions. This history is
updated via the Q-update
\begin{equation*}
    Q_{ki}(\tau + 1) = (1 - \alpha_k) Q_{ki}(\tau) + \alpha_k r_{ki}(\x_{-k}(\tau)),
\end{equation*}
where $\tau$ denotes the current time step. 

$Q_{ki}(\tau)$ denotes the \emph{Q-value} maintained
by agent $k$ about the performance of action $i \in S_k$. In effect, $Q_{ki}$ gives a discounted
history of the rewards received when $i$ is played, with $1 - \alpha_k$ as the discount factor.

Given these Q-values, each agent updates their mixed strategies according to the Boltzmann
distribution, given by
\begin{equation*}
    x_{ki}(\tau) = \frac{\exp(Q_{ki}(\tau)/T_k) }{\sum_j \exp(Q_{kj}(\tau)/T_k)},
\end{equation*}
in which $T_k \in [0, \infty)$ is the \emph{exploration rate} of agent $k$. 

It was shown in \cite{tuyls:qlearning,sato:qlearning} that a continuous time approximation of
the Q-Learning algorithm could be written as
\begin{equation} \tag{QLD} \label{eqn::QLD}
    \frac{\dot{x}_{k i}}{x_{k i}}=r_{k i}\left(\mathbf{x}_{-k}\right)-\langle \mathbf{x}_k, r_k(\mathbf{x}) \rangle +T_k \sum_{j \in S_k} x_{k j} \ln \frac{x_{k j}}{x_{k i}},
\end{equation}
which we call the \emph{Q-Learning dynamics} (QLD). The fixed points of this dynamic coincide
with the (QRE) of the game \cite{piliouras:zerosum}. QLD can also be seen as an entropy regularised form of the well-studied \emph{replicator dynamics} (RD) \cite{smith:replicator,hofbauer:book}. Besides its importance in the study of population biology \cite{chakraborty:chaos}, RD is known to be a special case of the generalised \emph{Follow the
Regularised Leader} learning dynamic \cite{mertikopoulos:reinforcement}, which models agents who
maximise their accumulated payoffs subject to a penalisation function. RD
has been shown to display asymptotic convergence in potential games \cite{hofbauer:book}, cyclic
behaviour in zero-sum games \cite{piliouras:cycles} and chaos in a number of other classes
\cite{sato:rps,griffin:evonetworks}. The connection between RD and QLD is explored in \cite{piliouras:potential}.

\section{Guaranteed Convergence of Q-Learning in Network Games}

\begin{figure*}[t]
    \captionsetup{justification=centering}
        \centering
        \begin{subfigure}[b]{0.3\textwidth}
            \centering
            \includegraphics[width=0.6\textwidth]{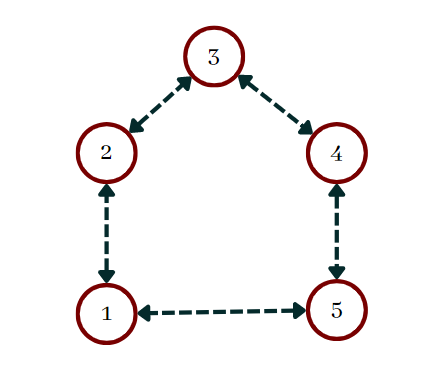}
            \caption*{Ring Network \\ $\left\lVert G \right\rVert_\infty = 2, \left\lVert G \right\rVert_2 = 2$}
        \end{subfigure}
        \begin{subfigure}[b]{0.3\textwidth}
            \centering
            \includegraphics[width=0.6\textwidth]{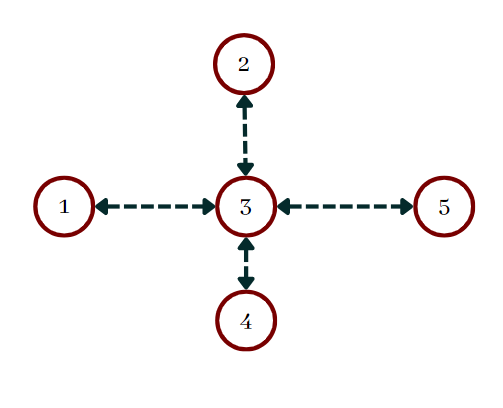}
            \caption*{Star Network \\ $\left\lVert G \right\rVert_\infty = N-1, \left\lVert G
            \right\rVert_2 = \sqrt{N - 1}$}
        \end{subfigure}
        \begin{subfigure}[b]{0.3\textwidth}
            \centering
            \includegraphics[width=0.6\textwidth]{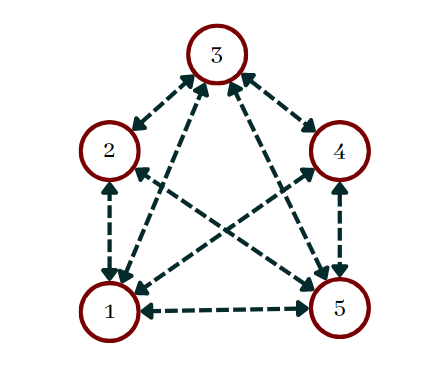}
            \caption*{Fully Connected Network \\ $\left\lVert G \right\rVert_\infty = N-1, \left\lVert G \right\rVert_2 = N - 1$}
        \end{subfigure}
    
     \caption{Examples of networks with five agents and associated $\left\lVert G \right\rVert_\infty$
     and $\left\lVert G \right\rVert_2$. }\label{fig::example-networks}
    \end{figure*}
    

    \begin{figure*}[t]
        \centering
        \begin{subfigure}[b]{0.45\textwidth}
            \centering
            \includegraphics[width=0.85\textwidth]{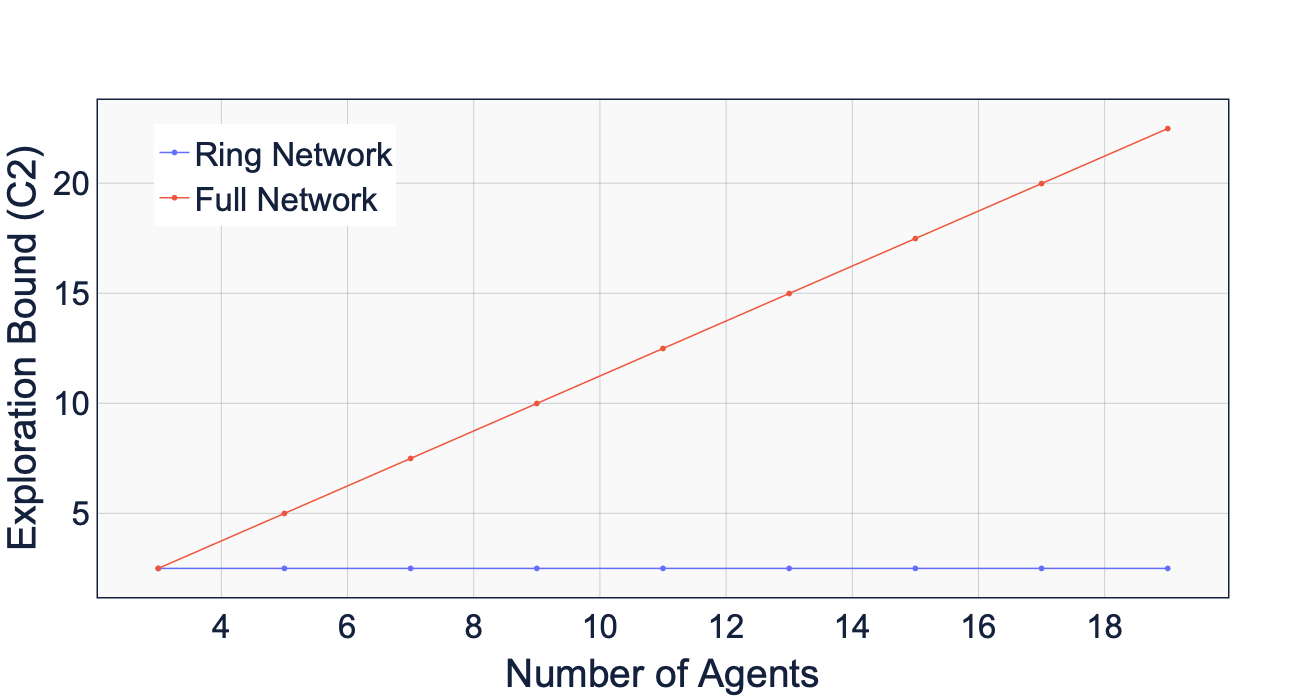}
            \caption*{Shapley Game}
        \end{subfigure}
        \begin{subfigure}[b]{0.45\textwidth}
            \centering
            \includegraphics[width=0.85\textwidth]{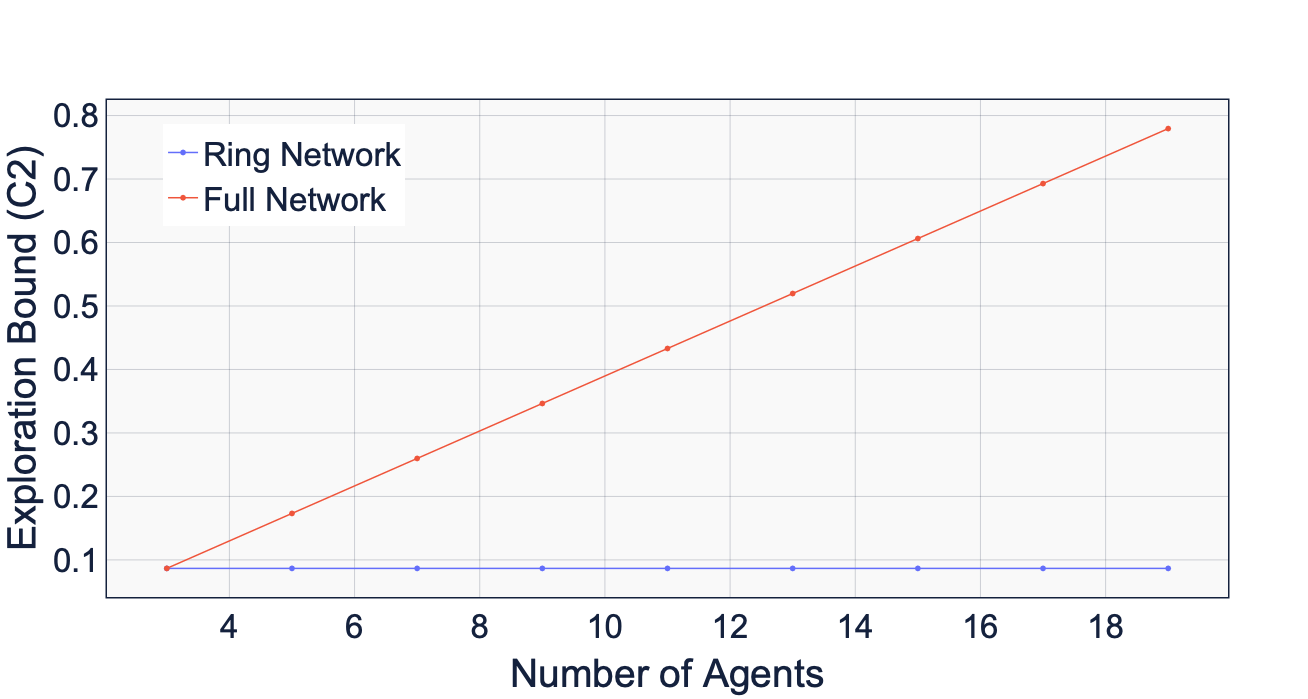}
            \caption*{Sato Game}
        \end{subfigure}
      
        \begin{subfigure}[b]{0.45\textwidth}
            \centering
            \includegraphics[width=0.85\textwidth]{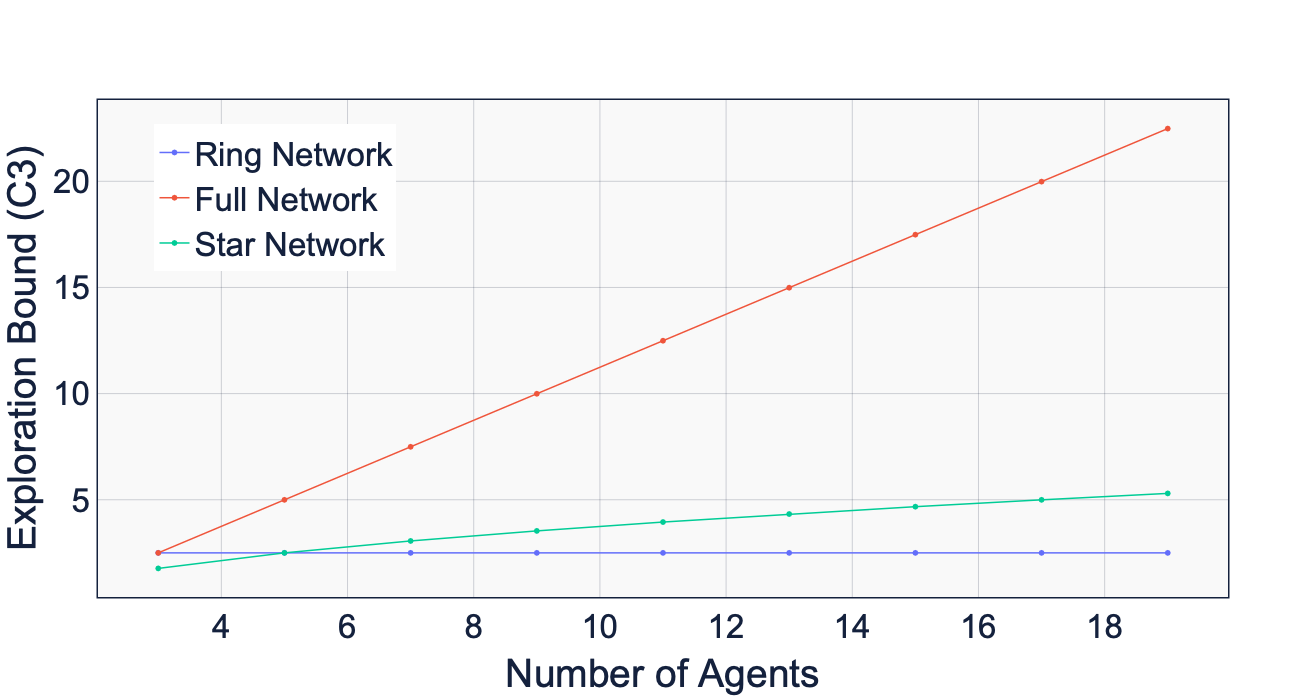}
            \caption*{Shapley Game}
        \end{subfigure}
        \begin{subfigure}[b]{0.45\textwidth}
            \centering
            \includegraphics[width=0.85\textwidth]{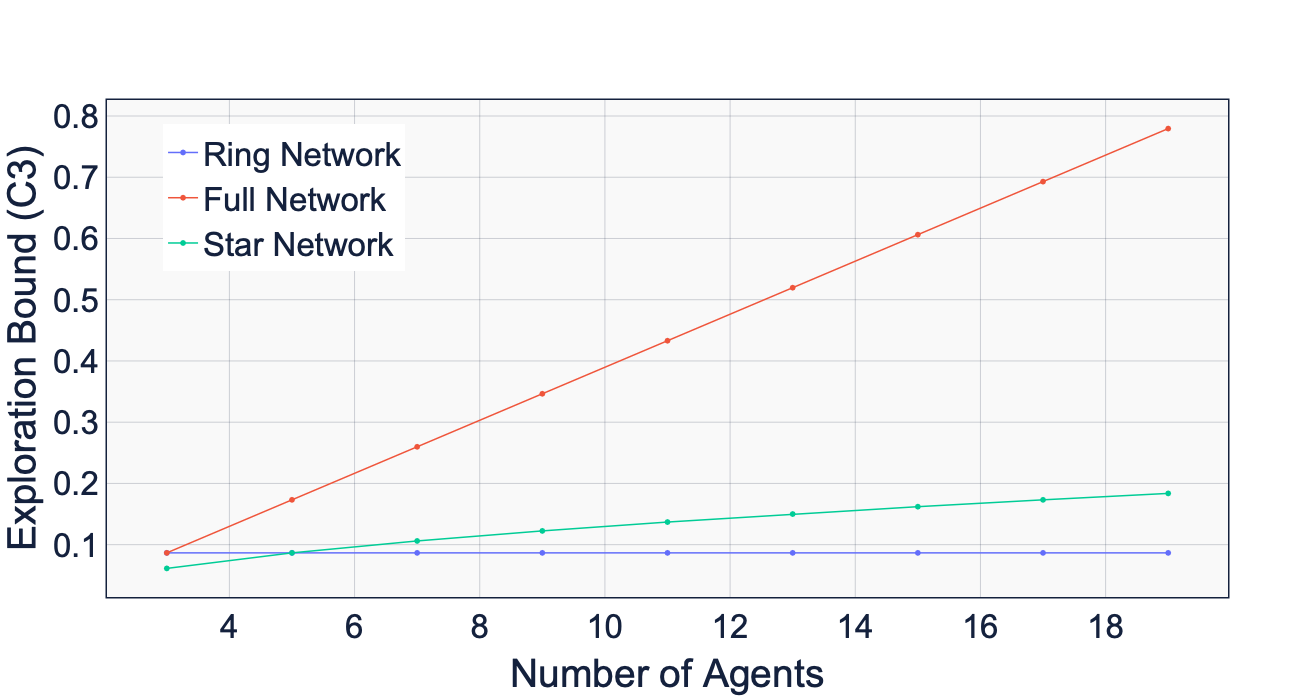}
            \caption*{Sato Game}
        \end{subfigure}
     \caption{Lower Bound on sufficient exploration as defined by (Top) (\ref{eqn::infty-cond})
     in a Full Network and Ring Network (Bottom) (\ref{eqn::2-cond}) in a Full Network, Star Network and Full
     Network.}\label{fig::stability-boundary}
    \end{figure*}

   \begin{figure*}[t]
        \centering
        \begin{subfigure}[b]{0.45\textwidth}
            \centering
            \includegraphics[width=0.85\textwidth]{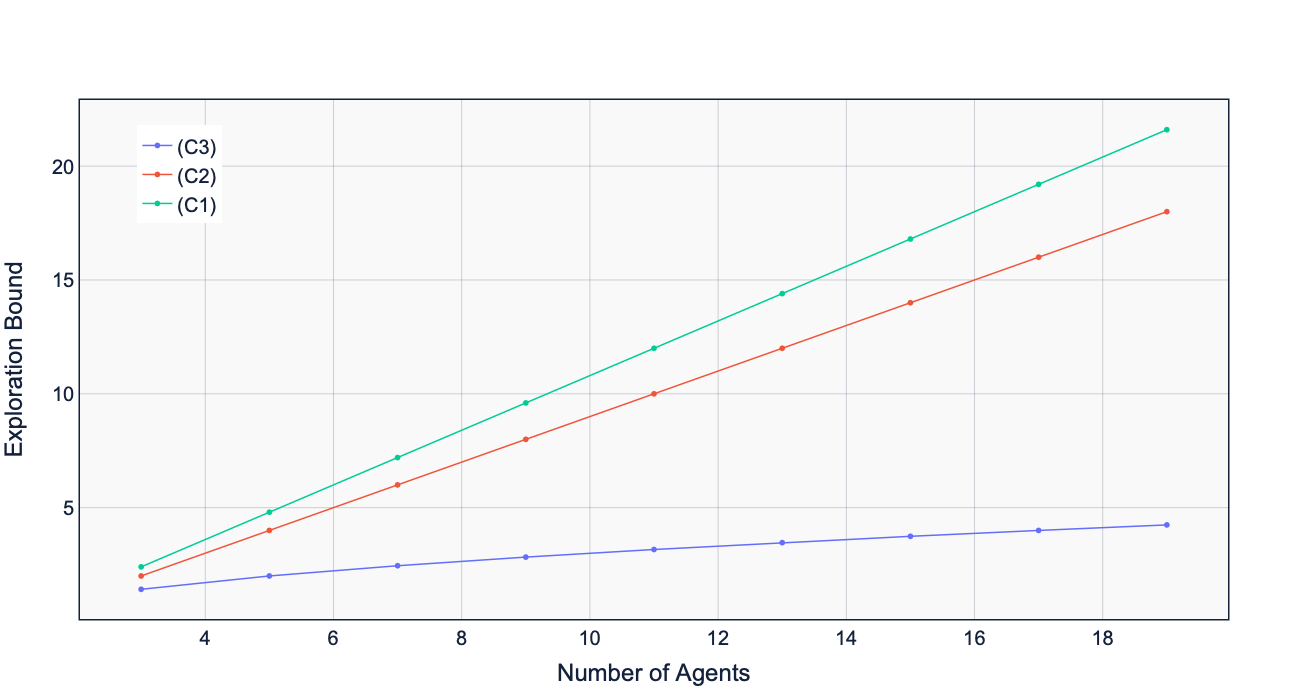}
            \caption*{Shapley Game}
        \end{subfigure}
        \begin{subfigure}[b]{0.45\textwidth}
            \centering
            \includegraphics[width=0.85\textwidth]{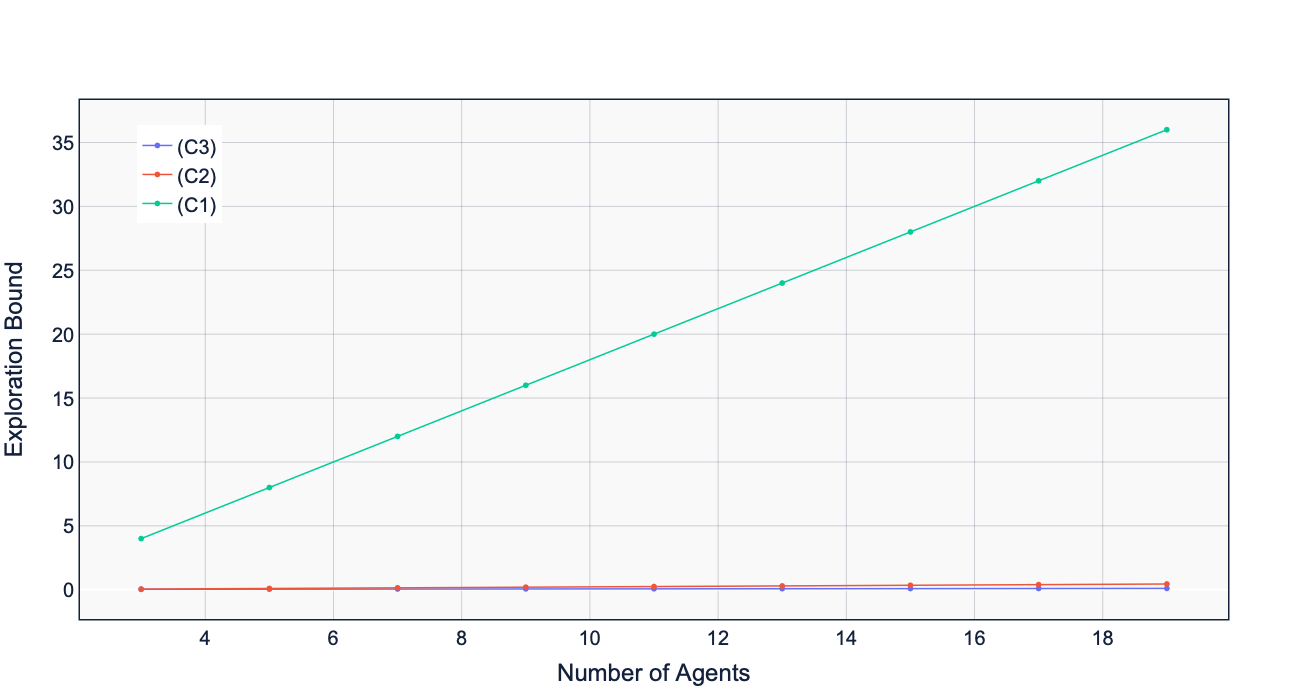}
            \caption*{Sato Game}
        \end{subfigure}
      
     \caption{Lower Bound on sufficient exploration as defined by (\ref{eqn::influence-cond}), (\ref{eqn::infty-cond}) and (\ref{eqn::2-cond}) in a Star Network. For (\ref{eqn::influence-cond}), $\max_{k \in \agentset} \delta_k |\agentset_k|$ is depicted which therefore coincides with the condition defined in \cite{hussain:aamas}.}\label{fig::cond-comparison}
    \end{figure*}

In this section we determine a number of sufficient conditions on the exploration rates $T_k$ under
which Q-Learning dynamics converge to a unique QRE. We find that these conditions are dependent on
the structure of the rewards in the game, parameterised by the interaction coefficient or the
inflence bound, and also on the structure of the network. We then compare our result to that of
\cite{hussain:aamas} and show that, under suitable network structures, stability can be achieved
with comparatively low exploration rates, even in the presence of many players. This also refines
the result of \cite{sanders:chaos} which suggests that learning dynamics are increasingly unstable
as the number of players increases, regardless of exploration rate. All proofs are in Appendix \ref{sec::main-thm-proof}.

\begin{theorem} \label{thm::main-thm} 
  Consider a network game $\game = (\agentset, \edgeset, (u_k, \actionset{k})_{k \in \agentset})$
  which has a network adjacency matrix $G$. Let $\sigma_I$ denote the intensity of identical interests for
  $\game$ and $\delta_k$ denote the influence bound of each agent $k \in \agentset$. Then, the
  Q-Learning Dynamic converges to a unique QRE $\NE \in \Delta$ if any of the following conditions
  hold for all agents $k \in \agentset$, 
  \begin{align}
      T_k &> \delta_k |\agentset_k|, \tag{C1} \label{eqn::influence-cond} \\
      T_k &> \frac{1}{2} \sigma_I \left\lVert G \right \rVert_\infty \tag{C2} \label{eqn::infty-cond},
  \end{align}
  where $\lVert M \rVert_\infty = \max_i \sum_{j} |[M]_{ij}|$ is the operator
  $\infty$-norm.
  If, in addition, each edge defines the same bimatrix game $(A, B)$, then asymptotic
  convergence of Q-Learning Dynamics holds if, for all $k \in \agentset$
  \begin{equation} \tag{C3} \label{eqn::2-cond}
      T_k > \frac{1}{2} \sigma_I \left\lVert G \right\rVert_2.
  \end{equation}
\end{theorem}


\begin{remark}
    Condition (\ref{eqn::influence-cond}) immediately refines the result of \cite{hussain:aamas} to the case of network games. In the latter work, the authors implicitly assume that the reward for each agent depends on all other agents. In our work, this corresponds exactly to the case of a fully connected network, where $\agentset_k = \agentset \backslash \{k\}$. In addition, \cite{hussain:aamas} define the influence bound to be over all agents, yielding a single condition which must hold for all $k$. Instead (\ref{eqn::influence-cond}) allows for agents who have a lower $\delta_k$ or who are not strongly connected in the network to have lower exploration rates $T_k$ without compromising convergence.
\end{remark}

\begin{remark}
    We can directly compare (\ref{eqn::influence-cond}) and (\ref{eqn::infty-cond}) due to the definition of the infinity norm. In particular $\lVert G \rVert_\infty = \max_k |\agentset_k|$ is the maximum number of neighbours for any agent $k \in \agentset$. Therefore, in a network where all agents are connected identically, the network dependency in (\ref{eqn::influence-cond}) is the same as that in (\ref{eqn::infty-cond}) . Next, the advantage of using the influence bound is that its definition applies in games which are not defined by matrices, and so the result generalises outside of network polymatrix games. By contrast, $\sigma_I$ is often easier to compute than $\delta_k$ as it is based on matrix norms rather than pairwise differences. Furthermore, $\frac{1}{2} \sigma_I$ is less than $\delta_k$ in a number of polymatrix games (c.f. Sec.~\ref{sec::experiments}). In summary, (\ref{eqn::influence-cond}) presents an advantage in terms of generality , whilst (\ref{eqn::infty-cond}) is often easier to compute and can be a tighter bound in network polymatrix games where all agents are identically coupled.
\end{remark}

\begin{remark}
    Theorem \ref{thm::main-thm} applies generally across all network polymatrix games, without making any assumptions, such as the network zero-sum condition. In fact, for networks of pairwise zero sum games, the following holds
\end{remark}

\begin{corollary}
        If the network game $\game$ is a pairwise zero-sum matrix, i.e.~$A^{kl} + (A^{lk})^\top = 0$
        for all $(k, l) \in \edgeset$, then the Q-Learning dynamics converge to a unique QRE so long as
        exploration rates $T_k$ for all agents are strictly positive.
\end{corollary}

Corollary 1 is supported by the result of \cite{piliouras:zerosum,hussain:aamas} in which it was
shown that Q-Learning converges to a unique QRE in all network zero-sum games, even if they are
not pairwise zero-sum , so long as all exploration rates $T_k$ are positive.

\begin{remark}
    Whilst (\ref{eqn::2-cond}) requires a stronger assumption, namely that
    each edge corresponds to the same bimatrix game, this setting is well
    motivated in the literature \cite{szabo:graphs,griffin:evonetworks}. In addition, it holds that $\lVert G \rVert_2 \leq \lVert G
    \rVert_\infty$ for all symmetric matrices $G$. Therefore,
    (\ref{eqn::2-cond}) provides a stronger bound than (\ref{eqn::infty-cond}). Figure \ref{fig::stability-boundary} depicts (\ref{eqn::infty-cond}) and (\ref{eqn::2-cond}) on various network games, whilst a direct comparison is visualised in Figure \ref{fig::cond-comparison}. 
\end{remark}

\subsection{QRE as approximate Nash Equilibria} \label{sec::e-NE}

In the following section we compare the QRE as an equilibrium solution to the Nash Equilibrium (NE) condition. In particular we show that the QRE of any game, which no longer needs to be a network game, is close to an NE in the following sense

\begin{definition}[$\epsilon$-approximate Nash Equilibrium] \label{def::e-NE}
    A strategy $\NE \in \Delta$ is an \emph{$\epsilon$-approximate Nash
    Equilibrium} for the game $\game$ if, for all agents $k$, and all strategies
    $\y_k \in \Delta_k$
    \begin{equation*}
        u_k(\y_k, \NE_{-k}) - u_k(\NE_k, \NE_{-k}) \leq \epsilon.
    \end{equation*}
\end{definition}

\begin{proposition} \label{prop::e-NE}
    Consider a game $\game$ and let $T_1, \ldots, T_N > 0$ denote positive
    exploration rates. Then any QRE $\NE \in \Delta$ is an
    $\epsilon$-approximate Nash Equilibrium where
    \begin{align}
        \epsilon &= \max_{k \in \agentset}  T_k A_k(\NE_k),  \label{eqn::e-NE} \\
        A_k(\x_k) &= \max_{i \in S_k} \ln x_{ki} - \langle \x_k, \ln \x_k \rangle \label{eqn::A_k}.
    \end{align} 
\end{proposition}

\begin{remark}
    Comparing (\ref{eqn::e-NE}) with (\ref{eqn::QLD}), it can be seen that $\epsilon$ denotes the maximum amount of entropy regularisation applied to the payoffs at the QRE $\NE$. Of course, this depends on the value of $\NE$ itself. As an example, if the QRE is the uniform distribution, i.e.~$\NE_k = (1/n_k, \ldots, 1/n_k)$ for all agents $k$, then $A_k(\NE_k) = 0$. In this case, $\NE$ is exactly an NE of the game.
\end{remark}

\begin{remark}
    It is also important to note that value of $\epsilon$ given by any QRE $\NE$ holds exactly. This gives the tightest possible approximation of Nash for any given QRE $\NE$. Whilst it is largely known that QRE can be considered as approximations of Nash \cite{turocy:homotopy,mckelvey:qre,gemp:sample}, to our knowledge Proposition \ref{prop::e-NE} is the first which exactly quantifies the `distance' between the two equilibrium concepts.
\end{remark}

We plot $A_k(\x)$ for the case $n_k = 3$ and $n_k = 2$ in the Appendix (Figure \ref{fig::surprisal_plot}). To determine its upper bounds, note that $A_k(\NE_k) \leq \max_{\x_k \in \Delta_k} A_k(\x_k) =: \bar{A}_k$. The form for $\bar{A}_k$ is in general unavailable in closed form and so we give exact values in the Appendix, focusing here on sharp bounds.

\begin{lemma}[Full version in Lemma \ref{lem::A_k-exact}] \label{lem::A_k-bound}
    \begin{align*}
        \Bar{A}_k &:= \max_{\x_k \in \Delta_k} \left(\max_{i \in S_k} \ln x_{ki} - \langle \x_k, \ln \x_k \rangle \right) = \mathcal{O}(\ln n_k).
    \end{align*}
\end{lemma}

\subsection{Updating Exploration Rates} \label{sec::iterative-scheme}

In this section, we use Theorem \ref{thm::main-thm} and Proposition \ref{prop::e-NE} to devise a scheme to update exploration rates so that which Q-Learning dynamics are driven `close' to a NE. The full algorithm is provided in the Appendix, with the main ideas discussed here. Starting with a choice of $T_k$ which satisfies any of the conditions in Theorem \ref{thm::main-thm}, it is clear that agents will achieve an $\epsilon$-NE where $\epsilon$ is given by (\ref{eqn::e-NE}). First, we notice that the value of $\epsilon$ depends only on the agent who maximises $T_k A_k(\NE_k)$. Therefore, it is natural to decrease the exploration rate for only this agent. We repeat this process until another agent maximises $T_k A_k(\NE_k)$, in which case this becomes the agent whose exploration rate is decreased, or the learning dynamics no longer achieve asymptotic convergence, at which point the learning process stops, and the last found QRE is chosen as the final joint strategy of all agents. To evaluate whether the system achieves asymptotic convergence for any choice of $T_k$, a window of the final $H$-iterations of learning is recorded and, for each $k \in \agentset$, $i \in \actionset{k}$ the relative difference between the maximum and minimum value of $x_{ki}$ across the window is determined. If this value is less some tolerance, the system is said to have converged. More formally the dynamics are said to have converged if
\begin{equation}\label{eqn::conv-criteria}
\left( \frac{\max_{t\in H} x_{ki}(t) - \min_{t \in H} x_{ki}(t)}{\max_{t \in H} x_{ki}(t)} \right) < l.
\end{equation}
By following this process, agents iteratively reach QRE which are closer approximations of an NE. We evaluate this process in our experiments and show that, even in large scale games, the $\epsilon$-approximation of the NE improves leading to optimal, and stable, learned joint strategies.

\section{Experiments} \label{sec::experiments}
\begin{figure*}[t]
	\centering
	\begin{subfigure}[b]{0.225\textwidth}
		\centering
		\includegraphics[width=0.8\textwidth]{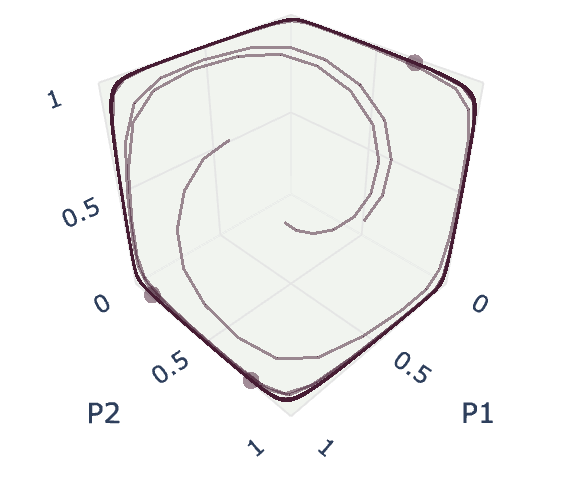}
		\caption*{$T = 0.7$}
	\end{subfigure}
	\begin{subfigure}[b]{0.225\textwidth}
		\centering
		\includegraphics[width=0.8\textwidth]{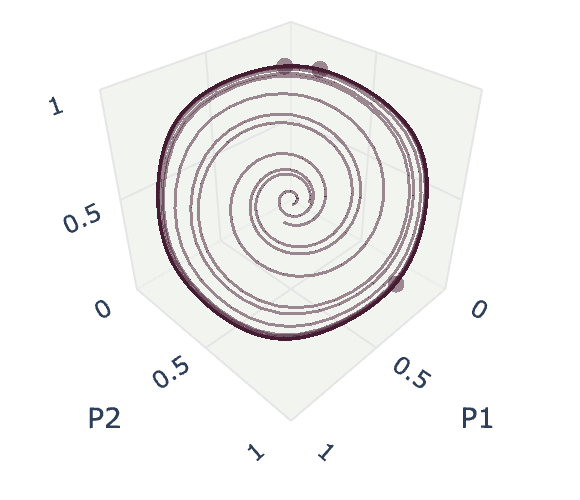}
		\caption*{$T = 1.5$}
	\end{subfigure}
	\begin{subfigure}[b]{0.225\textwidth}
		\centering
		\includegraphics[width=0.8\textwidth]{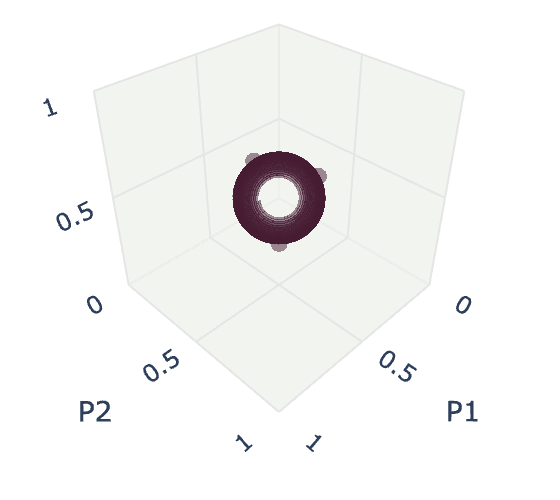}
		\caption*{$T = 2$}
	\end{subfigure}
	\begin{subfigure}[b]{0.225\textwidth}
		\centering
		\includegraphics[width=0.8\textwidth]{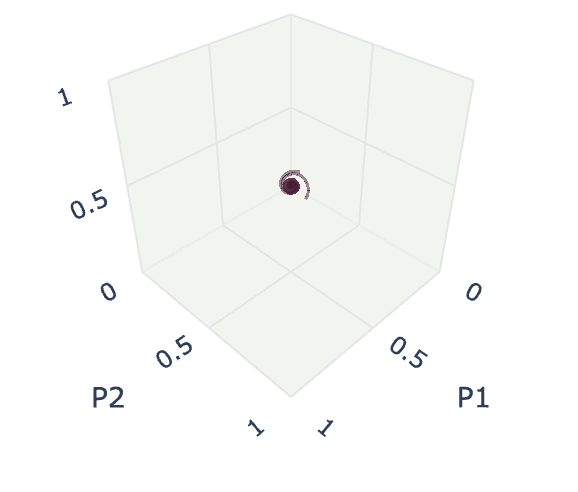}
		\caption*{$T = 2.7$}
	\end{subfigure}

 \caption{Trajectories of Q-Learning in a three agent  Network Chakraborty Game with $\alpha =
 7, \beta = 8.5$. Axes denote the probabilities with which
 each player chooses their first action.}\label{fig::chakraborty-traj}
\end{figure*}

\begin{figure*}[t!]
	\centering
	\begin{subfigure}[b]{0.45\textwidth}
		\centering
		\includegraphics[width=0.85\textwidth]{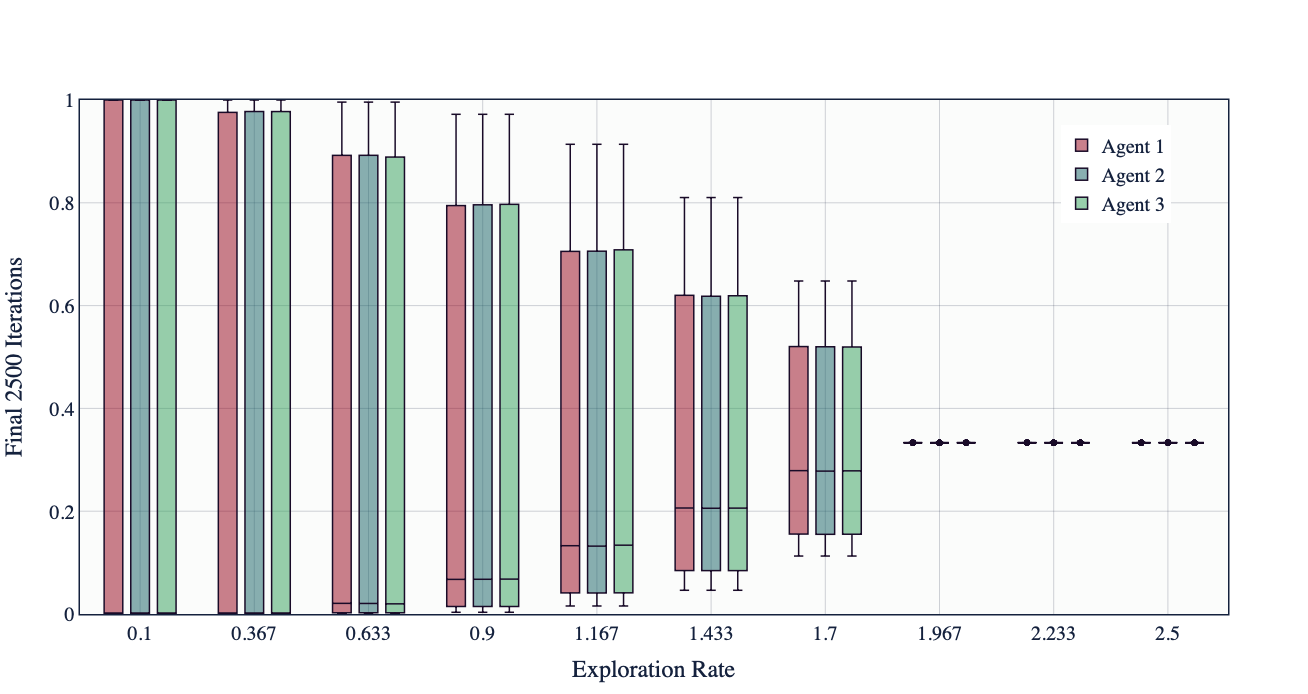}
		\caption*{Fully Connected Network}
	\end{subfigure}
	\begin{subfigure}[b]{0.45\textwidth}
		\centering
		\includegraphics[width=0.85\textwidth]{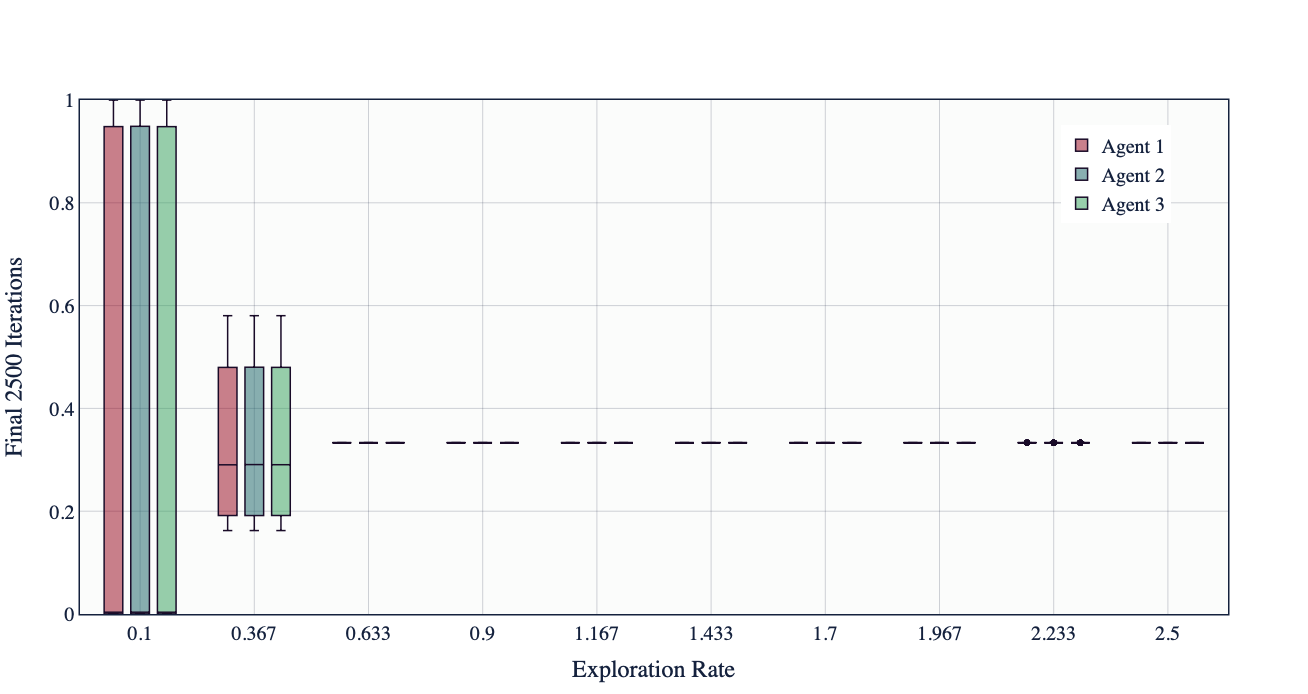}
		\caption*{Ring Network}
	\end{subfigure}
 	\begin{subfigure}[b]{0.45\textwidth}
		\centering
		\includegraphics[width=0.85\textwidth]{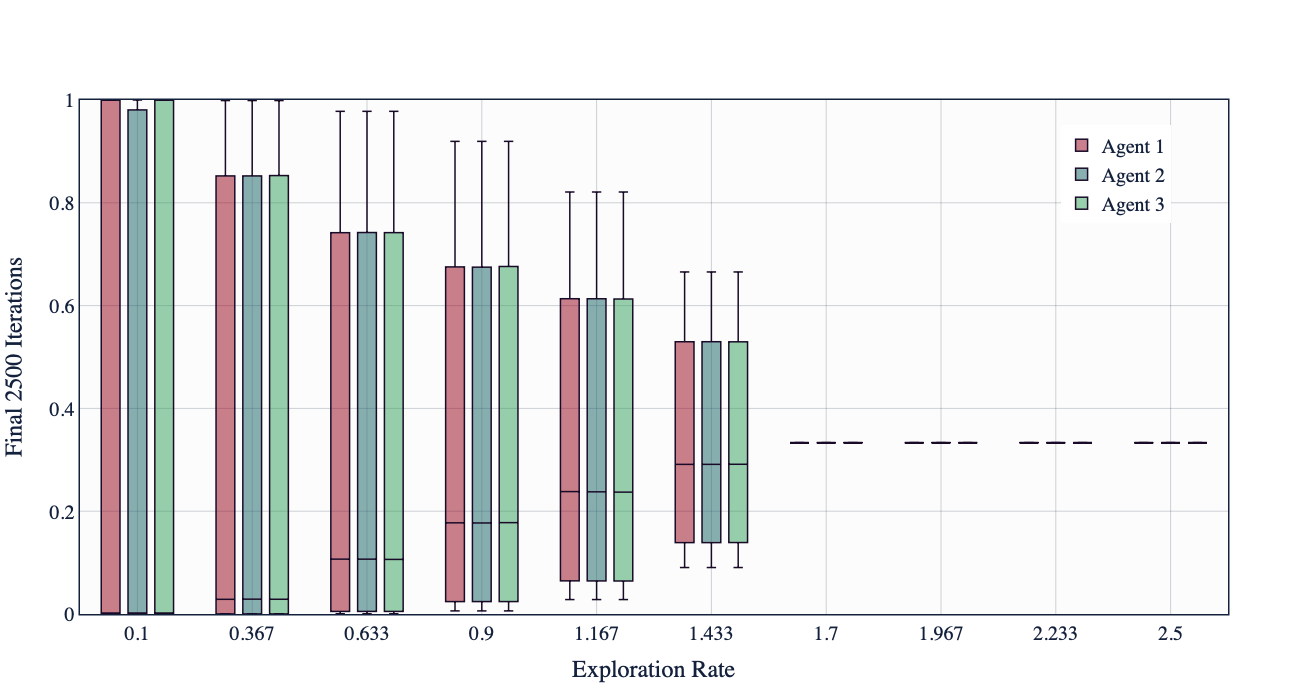}
		\caption*{Fully Connected Network}
	\end{subfigure}
	\begin{subfigure}[b]{0.45\textwidth}
		\centering
		\includegraphics[width=0.85\textwidth]{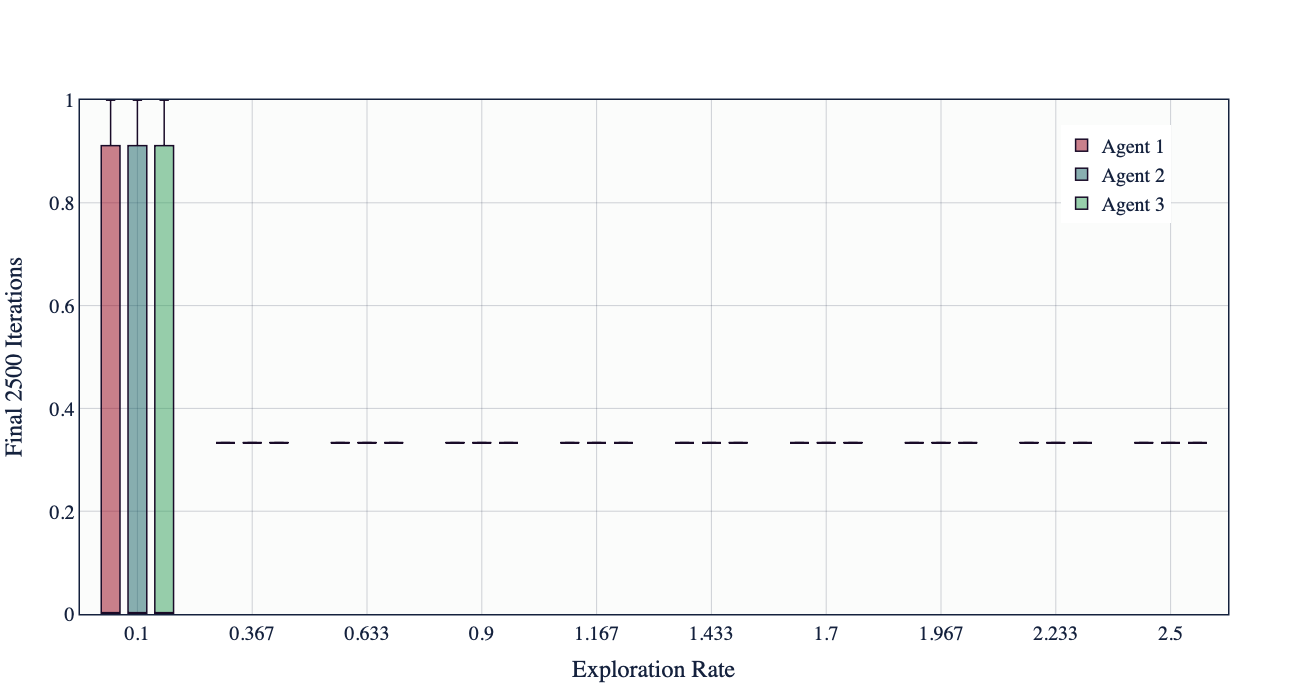}
		\caption*{Ring Network}
	\end{subfigure}
 \caption{Q-Learning in the (Top) Network Shapley Game (Bottom) Network Sato Game with 15 agents.
 The boxplot depicts the probabilities with which three of the agents play their first action in the
 final 2500 iterations of learning. This is depicted for varying choices of exploration rate $T$.}
 \label{fig::shapley-box}
\end{figure*}
\begin{figure*}[t!]
	\centering
	\begin{subfigure}[b]{0.45\textwidth}
		\centering
		\includegraphics[width=0.85\textwidth]{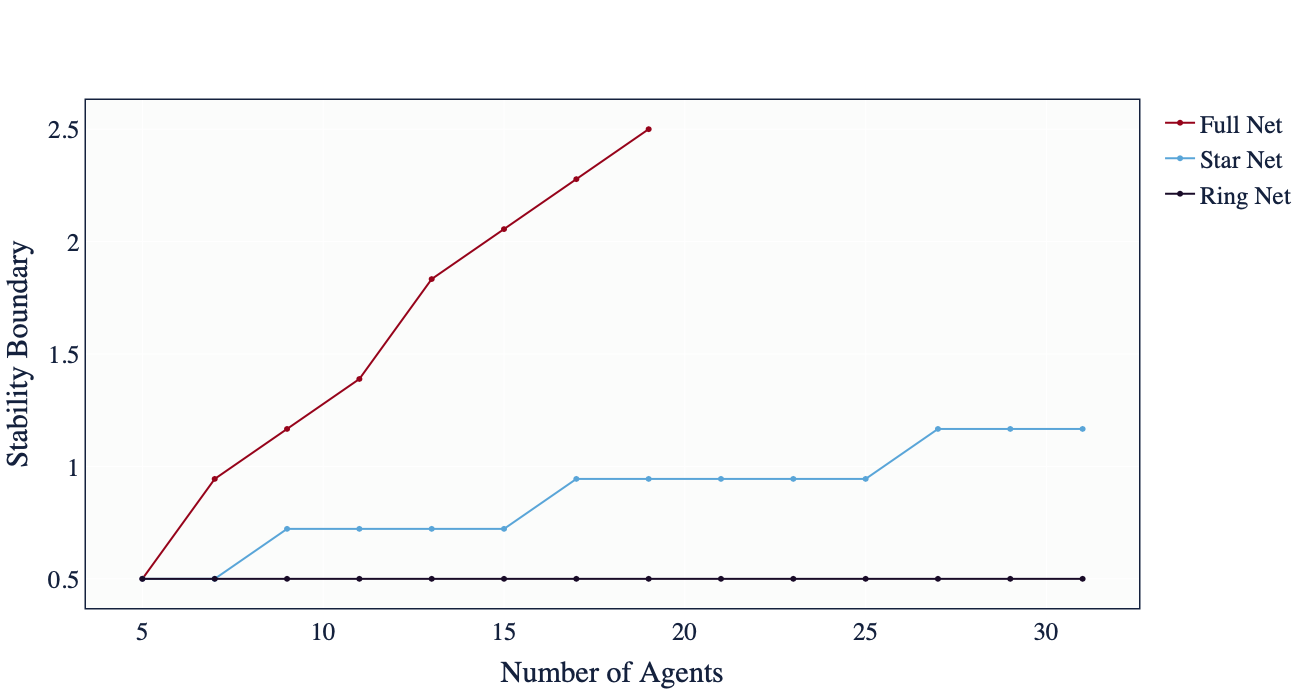}
		\caption*{Shapley Game}
	\end{subfigure}
	\begin{subfigure}[b]{0.45\textwidth}
		\centering
		\includegraphics[width=0.85\textwidth]{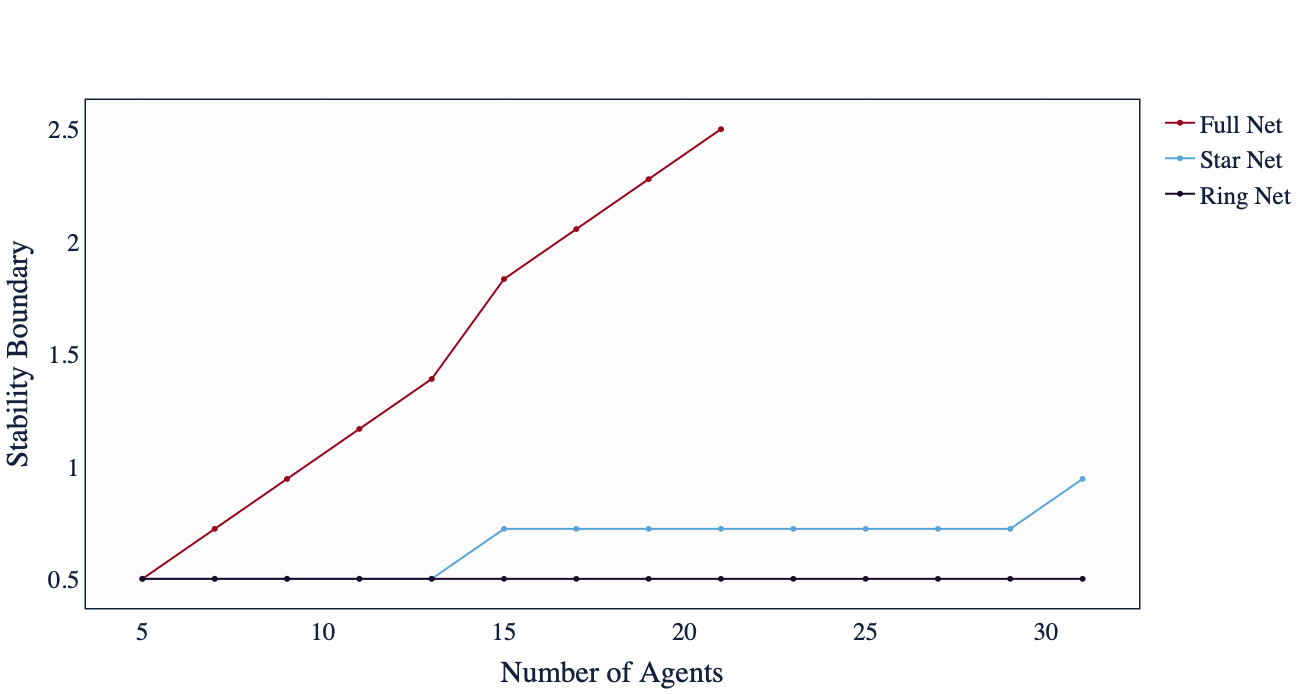}
		\caption*{Sato Game}
	\end{subfigure}
 \caption{Empirically determined stability boundary of Q-Learning measured against the number of
 agents. Q-Learning is iterated with 10 initial conditions and the game is considered to have
 converged if, for all agents and initial conditions (\ref{eqn::conv-criteria}) holds with $l = 1
 \times 10^{-5}$. The Fully Connected Network, Star Network and Ring Networks are considered.}
 \label{fig::empirical-boundary}
\end{figure*}

\begin{figure*}[t!]
	\centering


 	\begin{subfigure}[b]{0.45\textwidth}
		\centering
		\includegraphics[width=0.85\textwidth]{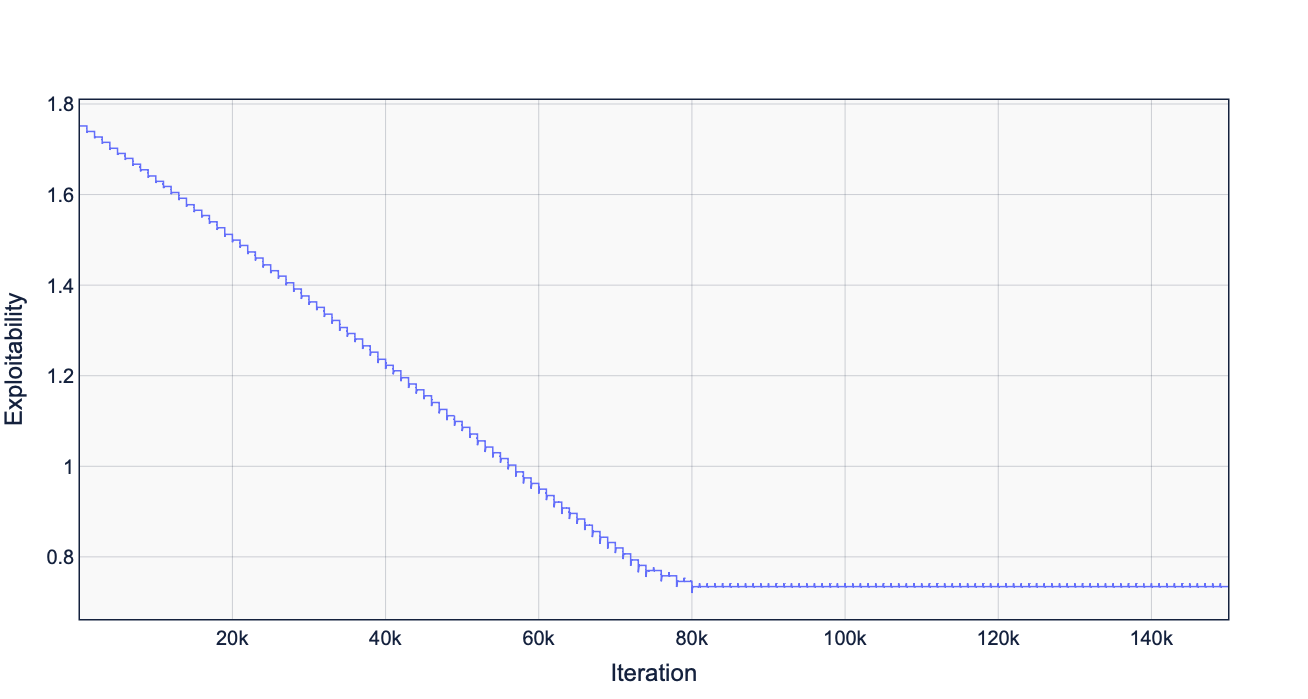}
	\end{subfigure}
	\begin{subfigure}[b]{0.45\textwidth}
		\centering
		\includegraphics[width=0.85\textwidth]{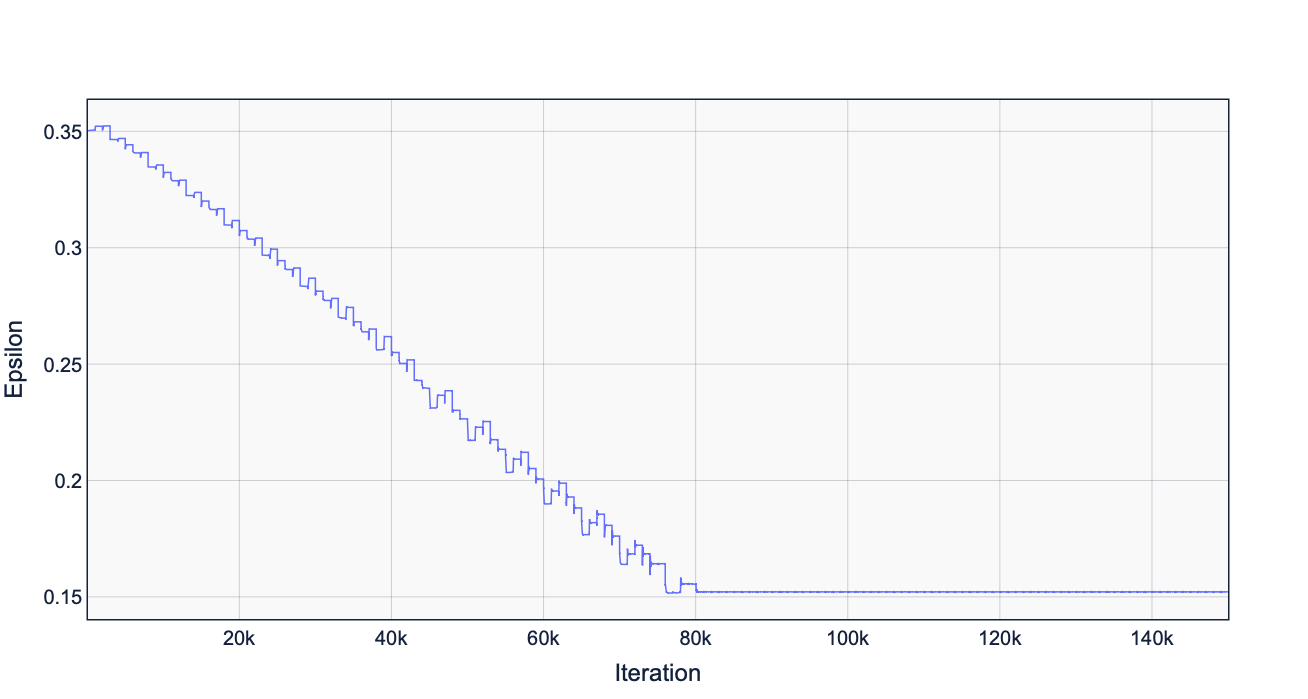}
	\end{subfigure}
 
 \caption{Measures of `closeness' to Nash Equilibrium as the exploration update scheme is applied to the Network Chakraborty Game with five agents and $\alpha=2.5, \beta=1.5$. (Left) Distance to NE measured by exploitability (\ref{eqn::exploitability}) of the joint strategy $\x(t)$. (Right) $\epsilon$ as defined by (\ref{eqn::e-NE}). Both metrics decreases as exploration rates are updated until condition (\ref{eqn::conv-criteria}) fails at approx. $8 \times 10^4$ iterations, after which learning is halted.}
 \label{fig::centralised-scheme}
\end{figure*}

\begin{figure}[t!]
	\centering
	\begin{subfigure}[b]{\columnwidth}
		\centering
		\includegraphics[width=0.85\columnwidth]{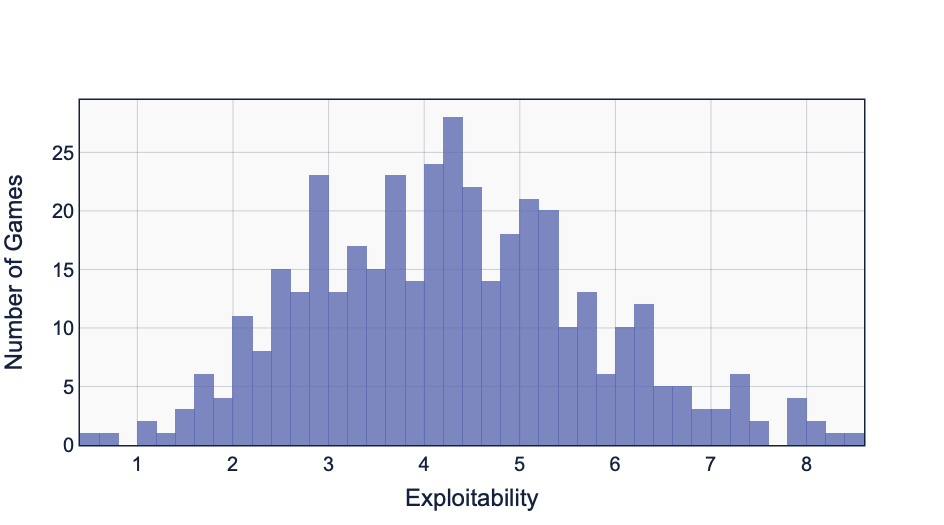}
	\end{subfigure}
	\begin{subfigure}[b]{\columnwidth}
		\centering
		\includegraphics[width=0.85\columnwidth]{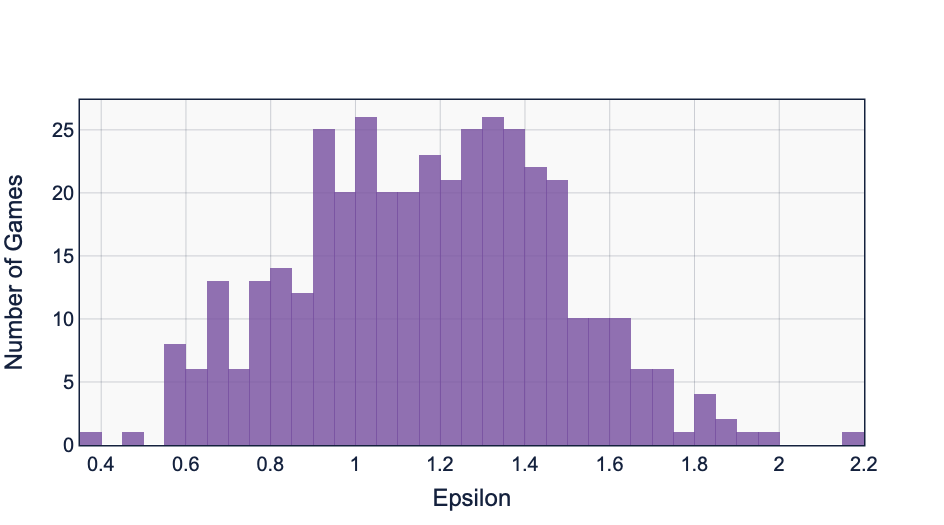}
	\end{subfigure}
 
 \caption{Histograms depicting the decrease of (Top) Exploitability and (Bottom) $\epsilon$ over $150,000$ iterations of learning across $500$ randomly generated network games with payoffs bounded in $[0, 5]$.}
 \label{fig::centralised-scheme-hist}
\end{figure}

We first visualise and exemplify the implications of our main result, Theorem \ref{thm::main-thm}, on a
number of games. In particular, we simulate the Q-Learning algorithm described in Section
\ref{sec::LearningModel} and show that Q-Learning asymptotically approaches a unique QRE so long as
the exploration rates are sufficiently large. We show, in particular, that the amount of exploration
required depends on the structure of the network rather than the total number of agents.

\begin{remark}
    In our experiments, we take all agents $k$ to have the same exploration rate $T$ and so drop the
    $k$ notation. As all bounds in Theorem \ref{thm::main-thm} must hold for all agents $k$, this assumption does not affect the generality of the results.
\end{remark}

\subsection{Convergence of Q-Learning} 

We first illustrate the convergence of Q-Learning using the \emph{Network Chakraborty Game}, which was analysed in \cite{chakraborty:deviations} to characterise chaos in
learning dynamics. Formally, the payoff to each agent $k$ is defined as
\begin{align*}
    &u_k(\x_k, \x_{-k}) = \x_k^\top \A \x_l, \; l = k-1 \mod N, \\
    &A = \begin{pmatrix}
        1 & \alpha \\ \beta & 0
    \end{pmatrix}, \; \alpha, \beta \in \R.
\end{align*}

We visualise the trajectories generated by running Q-Learning in Figure \ref{fig::chakraborty-traj}
for a three agent network and choosing $\alpha = 7, \beta = 8.5$. It can be
seen that, for low exploration rates, the dynamics reach a limit cycle around the boundary of the
simplex. However, as exploration increases, the dynamics are eventually driven towards a fixed point
for all initial conditions. 

\paragraph{Network Shapley Game} In the following example, each edge of the network game has
associated the same pair of matrices $A, B$ where
\begin{eqnarray*} 
A=\begin{pmatrix} 1 & 0 & \beta \\
            \beta & 1 & 0 \\ 0 & \beta & 1 \end{pmatrix}, \, B=\begin{pmatrix} -\beta & 1 & 0 \\ 0 &
            -\beta & 1 \\ 1 & 0 & -\beta \end{pmatrix},
\end{eqnarray*}    
where $\beta \in (0, 1)$.

This has been analysed in the two-agent case in \cite{shapley:twoperson}, where it was shown that
the \emph{Fictitious Play} learning dynamic do not converge to an equilibrium. \cite{hussain:aamas}
analysed the network variant of this game for the case of a ring network and numerically showed that
convergence can be achieved by Q-Learning through sufficient exploration. 

In Figure \ref{fig::shapley-box} we examine both a fully connected network and a ring network with 15 agents. In this case, the dynamics evolve in $\R^{45}$ which prohibits
a visualisation of the complete dynamics. To resolve this, we instead take three representative agents and depict the \emph{spread} of their strategies in the final 2500 iterations of learning. A bar which stretches from $0$ to $1$ indicates that the dynamics are spread across the simplex which may occur in a limit cycle or chaotic orbit that approaches the boundary of the simplex (c.f. Figure \ref{fig::chakraborty-traj}). 
These are seen to occur for low exploration rates. By contrast, when exploration rates are increased beyond a certain threshold, a flat line is seen which indicates that the dynamics are stationary, i.e. a fixed point has been reached. 
Importantly, the boundary at which equilibrium behaviour
occurs is higher in the fully connected network, where $\left\lVert G \right\rVert_\infty = 14$ than
in the ring network, where $\left\lVert G \right\rVert_\infty = 2$. This indicates that larger numbers of agents may be introduced in the environment without impacting stability, so long as a weakly connected network is chosen. 

\paragraph{Network Sato Game} We also analyse the behaviour of Q-Learning in a variant of the game
introduced in \cite{sato:rps}, where it was shown that chaotic behaviour is exhibited by learning
dynamics in the two-agent case. We extend this towards a network game by associating each edge with
the payoff matrices $A, B$ given by
\begin{eqnarray*} 
A=\begin{pmatrix} \epsX & -1 & 1 \\
1 & \epsX & -1 \\
-1 & 1 & \epsX \end{pmatrix}, \, B=\begin{pmatrix} \epsY & -1 & 1 \\
1 & \epsY & -1 \\
-1 & 1 & \epsY \end{pmatrix},
\end{eqnarray*}   
where $\epsX, \epsY \in \R$. Notice that for $\epsX = \epsY = 0$, this corresponds to the classic
Rock-Paper-Scissors game which is zero-sum so that, by Corollary 1, Q-Learning will converge to an
equilibrium with any positive exploration rates. We choose $\epsX=0.01, \epsY=-0.05$ in order to
stay consistent with \cite{sato:rps} which showed chaotic dynamics for this choice. The boxplot once
again shows that sufficient exploration leads to convergence of all initial conditions. However, the
amount of exploration required is significantly smaller than that of the Network Shapley Game. This
can be seen as being due to the significantly lower interaction coefficient of the Sato game
$\sigma_I = 0.05$ as compared to the Shapley game $\sigma_I = 2$.

\subsection{Stability Boundary} 

In these experiments we empirically determine the dependence of the
stability boundary w.r.t.~the number of agents. For accurate comparison with Figure
\ref{fig::stability-boundary}, we consider the Network Sato and Shapley Games in a fully-connected
network, star network and ring network. We iterate Q-Learning for various values of $T$ and
determine whether the dynamics have converged. To evaluate convergence, we apply (\ref{eqn::conv-criteria}) with $|H|=2500$ iterations and $l = 1\times10^{-5}$. In Figure \ref{fig::empirical-boundary}, we plot the smallest exploration rate $T$ for which (\ref{eqn::conv-criteria}) holds for varying choices of $N$. It can be
seen that the prediction of Theorem \ref{thm::main-thm} holds, in that the number of agents plays no
impact for the ring network whereas the increase in the fully-connected network is linear in $N$. In
addition, it is clear that the stability boundary increases slower in the Sato game than in the
Shapley game, owing to the smaller interaction coefficient. 

An additional point to note is that the stability boundary for the star network increases slower
than the fully-connected network in all games. We anticipate that this is due to the fact that the
$2$-norm $\lVert G \rVert_2$ in the star network is smaller than that of the fully-connected network
(c.f. Figure \ref{fig::example-networks}).

\subsection{Effectiveness of Exploration Update Scheme}

In these experiments, we evaluate the exploration update scheme outlined in Section \ref{sec::iterative-scheme}. using $|H|=500$ and $l = 1 \times 10^{-5}$. In Figure \ref{fig::centralised-scheme} we consider the Network Chakraborty Game with $\alpha=2.5, \beta=1.5$ We measure the `distance' between the strategy $\x(t)$ and the NE using two metrics: first by $\epsilon$ as given in (\ref{eqn::e-NE}) and second through \emph{exploitability} $\phi(\x) $ given as
\begin{equation} \label{eqn::exploitability}
    \phi(\x) = \sum_k \max_{\y_k \in \Delta_k} u_k(\y_k, \x_{-k}) - u_k(\x_k, \x_{-k}).
\end{equation}
The exploitability is used, sometimes under different names, as a measure of distance to the NE \cite{gemp:sample,perrin:fp} and, from (\ref{eqn::e-NE}) it can be seen that $\phi(\NE) = \sum_k T_k A_k(\NE_k)$ for any QRE $\NE$. The reason for examining $\phi$ is that its definition holds for any strategy $\x \in \Delta$, whilst (\ref{eqn::e-NE}) only holds at a QRE $\NE \in \Delta$. It can be seen in all cases that both metrics decrease as agents learn, until condition (\ref{eqn::conv-criteria}) is no longer satisfied. To examine the generality of this performance, we evaluate the exploration update scheme in 500 randomly generated network games with 15 agents, two actions and a ring structure. Exploitability and $\epsilon$ are evaluated at the first iteration and final iteration and the difference is recorded. Figure \ref{fig::centralised-scheme-hist} plots the decrease of both metrics as a histogram across all 500 games. These experiments (as well as additional presented in Appendix \ref{sec::add-expt}) suggest that, if exploration rates are updated according the scheme in Section \ref{sec::iterative-scheme}, independent learning agents may learn stable equilibrium strategies which closely approximate Nash Equilibria. 


\section{Conclusion}

In this paper we show that the Q-Learning dynamics is guaranteed to converge in arbitrary network
games, independent of any restrictive assumptions such as network zero-sum or potential. This allows
us to make a branching statement which applies across all network games.

In particular, our analysis shows that convergence of the Q-Learning dynamics can be achieved
through sufficient exploration, where the bound depends on the pairwise interaction between agents
and the structure of the network. Overall, compared to the literature, we are able to tighten the
bound on sufficient exploration and show that, under certain network interactions, the bound does
not increase with the total number of agents. This allows for stability to be guaranteed in network
games with many players.

A fruitful direction for future research would be to capture the effect of the payoffs through a
tighter bound than the interaction coefficient and to explore further how properties of the network
affect the bound. In addition, whilst there is still much to learn in the behaviour of Q-Learning in
stateless games, the introduction of the state variable in the Q-update is a valuable next step.



\begin{acks}
Aamal Hussain and Francesco Belardinelli are partly funded by the UKRI Centre for Doctoral Training in Safe and Trusted Artificial Intelligence (grant number EP/S023356/1). Dan Leonte acknowledges support from the EPSRC Centre for Doctoral Training in Mathematics of Random Systems: Analysis, Modelling and Simulation (EP/S023925/1).
This research was supported in part by the National Research Foundation, Singapore and DSO National Laboratories under its AI Singapore Program (AISG Award No: AISG2-RP-2020-016), grant PIESGP-AI-2020-01, AME Programmatic Fund (Grant No.A20H6b0151) from A*STAR.
\end{acks}



\bibliographystyle{ACM-Reference-Format} 
\balance
\bibliography{references}


\newpage
\onecolumn
\appendix
\section{Preliminaries}

In this section we outline the various tools and properties that we will use in
our proofs. 

\subsection{Variational Inequalities and Monotone Games}
Our aim in this work is to analyse the Q-Learning dynamics in network games without invoking any
particular structure on the payoffs (e.g.~zero-sum). To do this, we employ the \emph{Variational
Inequality} approach, which has been successfully applied towards the analysis of network games
\cite{melo:network,parise:network} as well as learning in games
\cite{mertikopoulos:finite,sorin:composite,hussain:aamas}. 

\begin{definition}[Variational Inequality] 
     Consider a set $\X \subset \R^d$ and a map $F \, : \X \rightarrow \R^d$. The Variational
     Inequality (VI) problem $VI(\X, F)$ is given as
        \begin{equation}\label{eqn::VIdef}
            \langle \x - \NE, F(\NE) \rangle \geq 0, \hspace{0.5cm} \text{ for all } \x \in \X.
        \end{equation}
    We say that $\NE \in \X$ belongs to the set of solutions to a variational inequality problem
    $VI(\X, F)$ if it satisfies (\ref{eqn::VIdef}).
\end{definition}

The premise of the variational approach to game theory \cite{facchinei:VI, Rosen1965ExistenceGames}
is that the problem of finding equilibria of games can be reformulated as determining the set of
solutions to a VI problem. This is done by choosing associating the set $\X$ with $\Delta$ and the
map $F$ with the \emph{pseudo-gradient} of the game.

\begin{definition}[Pseudo-Gradient Map]
    The pseudo-gradient map of a game $\game = (\agentset, \edgeset, (u_k, \actionset{k})_{k \in
\agentset})$ is given by $F(\x) = (F_k(\x))_{k \in \agentset} = (-D_{\x_k} u_k(\x_k, \x_{-k}))_{k
\in \agentset}$.
\end{definition}
The advantage of this formulation is that we can apply results from the study of Variational
Inequalities to determine properties of the game. These results rely solely on the form of the
pseudo-gradient map and so can generalise results which assume a potential or zero-sum structure of
the game \cite{hussain:aamas,kadan:exponential}. 
\begin{lemma}[\cite{melo:qre}]
    Consider a game $\game = (\agentset, \edgeset, (u_k, \actionset{k})_{k \in \agentset})$ and for
    any $T_1, \ldots, T_N > 0$, let $F$ be the pseudo-gradient map of $\game^H$. Then $\NE \in
    \Delta$ is a QRE of $\game$ if and only if $\NE$ is a solution to $VI(\Delta, F)$.
\end{lemma}
With this correspondence in place, we can analyse properties of the pseudo-gradient map and its
relation to properties of the game and the learning dynamic. One important property is
\emph{monotonicity}.
\begin{definition}
    A map $F \, : \X \rightarrow \R^d$ is
    \begin{enumerate}
        \item \emph{Monotone} if, for all $\x, \y \in \X$,
        \begin{equation*}
            \langle F(\x) - F(\y), \x - \y \rangle \geq 0.
        \end{equation*}
        \item \emph{Strongly Monotone} with constant $\alpha > 0$ if, for all $\x, \y \in \X$,
        \begin{equation*}
            \langle F(\x) - F(\y), \x - \y \rangle \geq \alpha ||\x - \y||^2_2.
        \end{equation*}
    \end{enumerate}
\end{definition}
\begin{definition}[Monotone Game]
    A game $\game$ is \emph{monotone} if its pseudo-gradient map is monotone.
\end{definition}
A large part of our analysis will be in determining conditions under which the pseudo-gradient map
is monotone. Upon doing so, we are able to employ the following results.
\begin{lemma}[\cite{melo:qre}] \label{lem::unique-qre} Consider a game $\game = (\agentset,
    \edgeset, (u_k, \actionset{k})_{k \in \agentset})$ and for any $T_1, \ldots, T_N > 0$, let $F$
    be the pseudo-gradient map of $\game^H$. $\game$ has a unique QRE $\NE \in \Delta$ if $F$ is
    strongly monotone with any $\alpha>0$.
\end{lemma}
\begin{lemma}[\cite{hussain:aamas}] \label{lem::ql-conv} If the game $G$ is \emph{monotone}, then
    the Q-Learning Dynamics (\ref{eqn::QLD}) converge to the unique QRE with any positive
    exploration rates $T_1, \ldots, T_N > 0$.
\end{lemma}
Finally, recall that an operator $f : \X \subset \R^n \rightarrow \R^n$ is \emph{strongly convex}
with constant $\alpha$ if, for all $\x, \y \in \X$
\begin{equation*}
    f(\y) \geq f(\x) + Df(\x)^\top (\y - \x) + \frac{\alpha}{2}\lVert\x - \y \rVert^2_2.
\end{equation*}
It is known that, if $f(\x)$ is strongly convex, then its Hessian $D^2_{\x} f(\x)$ is strongly
positive definite with constant $\alpha$. Thus, all eigenvalues of $D^2_{\x} f(\x)$ are larger than
$\alpha$. To apply this in our setting, we use the following result.
\begin{proposition}[\cite{melo:qre}] \label{prop::strong-conv} The function $f(\x_k) = T_k \langle
        \x_k, \ln \x_k \rangle$ is strongly convex with constant $T_k$.
\end{proposition}

\subsection{Matrix Norms}

In addition, the following definitions and properties hold for any matrix $A$.
\begin{enumerate}
    \item $\lVert A \rVert_2 = \sqrt{\lambda_{\max}(A^\top A)}$ where $\lambda_{\max}$ is the
    largest eigenvalue of $A$,
    \item $\lVert A \rVert_\infty = \max_i \sum_{j} |[A]_{ij}|$,
    \item $\rho(A) = \max_i |\lambda_i(A)|$ where $\lambda_i(A)$ denotes an eigenvalue of $A$.
\end{enumerate}
\begin{proposition}[Weyl's Inequality] \label{prop::weyl} Let $J = D + N$ where $D$ and $N$ are
        symmetric matrices. Then it holds that
        \begin{equation*}
            \lambda_{\min}(J) \geq \lambda_{\min} (D) + \lambda_{\min} (N).
        \end{equation*}
        where $\lambda_{\min}(A)$ denotes the smallest eigenvalue of a matrix. 
\end{proposition}
\begin{proposition} \label{prop::kron} Let $G, A$ be matrices and $\otimes$ denote the Kronecker
    product. Then \begin{equation*} \lVert G \otimes A \rVert_2 = \lVert G \rVert_2 \lVert A
    \rVert_2. \end{equation*} \end{proposition}
\begin{proposition} \label{prop::spectral-radius} Let $A$ be a symmetric matrix. Then
    \begin{equation*}
       |\lambda_{\min}(A)| \leq \rho(A) = \lVert A \rVert_2.
       \end{equation*}
\end{proposition}
The following result is used in our proof to be able to parameterise the effect of pairwise
interactions by $\sigma_I$.
\begin{lemma} \label{lem::two-norm} Let $G \in \mathcal{M}_{N}(\R)$ be matrix for which each entry
    $g_{ij} \defeq \left[G\right]_{ij}$ is either $0$ or $1$. Let $N \in \mathcal{M}_{Nn}(\R)$ be a
    block matrix such that
    \begin{equation*}
        \left[N\right]_{ij} = \begin{cases}
            A^{ij} & \text{ if } g_{ij} = 1 \\ 
            \zeros & \text{ otherwise}
        \end{cases},
    \end{equation*}
    where $A^{ij} \in \mathcal{M}_n(\mathbb{R})$ are matrices of the same dimension. 
    let $A \in M_n(\R)$ be a matrix which satisfies $\lVert A \rVert_2 \geq \lVert B^{ij} \rVert_2$
    for all $(i, j) \in \edgeset$. Finally let $\Tilde{N} \in M_{Nn}(\R)$ be a block matrix given by
    Then 
    \begin{equation*}
 \left\lVert N \right\rVert_2 \leq \sqrt{\left\lVert G \right\rVert_1 \left\lVert G \right\rVert_\infty}\max_{1 \le i,j \le n} \left\lVert A_{ij} \right\rVert_2.
    \end{equation*}
\end{lemma}
\begin{proof} Let $v = (v^1,\ldots,v^n)\in \mathbb{R}^{Nn}$ where $v^{i} \in \mathbb{R}^N$ for $1
\le i \le n$. Then
\begin{equation}\lVert Nv \rVert_2^2 =   \left\lVert\begin{pmatrix}
            g_{11}A^{11} & \ldots  & g_{1n}A^{1n} \\
            \vdots &  &   \vdots \\
             g_{n1}A^{n1} & \ldots & g_{nn}A^{nn} 
        \end{pmatrix} \begin{bmatrix}
            v^1\\ \vdots  \\ v^n
        \end{bmatrix}\right\rVert_2^2 =\left\lVert \begin{bmatrix}
          \sum_{1j} g_{1j} A^{1j} v^j \\ \vdots \\ \sum_{ni} g_{nj}A^{nj} v^j
    \end{bmatrix}\right\rVert_2^2 \leq \sum_{i=1}^n \left\lVert \sum_{j=1}^n g_{ij} A^{ij}v^j  \right\rVert_2^2.\label{eq:squared_norm_of_nv}
        \end{equation}
For each fixed $i \in \{1,\ldots,n\}$, we have the upper bound
\begin{equation}
\left\lVert \sum_{j=1}^n g_{ij} A^{ij}v^j  \right\rVert_2 \le \sum_{j=1}^n g_{ij} \left\lVert A^{ij}v^j  \right\rVert_2 \le \sum_{j=1}^n g_{ij} \left\lVert A^{ij}  \right\rVert_2 \left\lVert v^j \right\rVert_2 \le \max_{1 \le i,j \le n}\left\lVert A^{ij}  \right\rVert_2 \sum_{j=1}^n g_{ij}  \left\lVert v^j \right\rVert_2.\label{eq:bound_for_each_i}
\end{equation}
By plugging \eqref{eq:bound_for_each_i} in \eqref{eq:squared_norm_of_nv} and expanding the squared
bracket, we obtain that
\begin{align*}
    \lVert Nv \rVert_2^2 \le \sum_{i=1}^n \left(\max_{1 \le i,j \le n}\left\lVert A^{ij}  \right\rVert_2 \sum_{j=1}^n g_{ij}  \left\lVert v^j \right\rVert_2 \right)^2 =& \max_{1 \le i,j \le n}\left\lVert A^{ij}  \right\rVert_2^2 \sum_{i=1}^n \sum_{k,l=1}^n g_{ik} g_{il} \left\lVert v^k \right\rVert_2 \left\lVert v^l \right\rVert_2\\ 
    \le & \max_{1 \le i,j \le n}\left\lVert A^{ij}  \right\rVert_2^2 \sum_{i=1}^n \sum_{k,l=1}^n g_{ik} g_{il} \left(\frac{1}{2}\left\lVert v^k \right\rVert_2^2 + \frac{1}{2}\left\lVert v^l \right\rVert_2^2 \right),
\end{align*}
where the last inequality follows by completing the square. Notice that the two sums above are
identical, hence
\begin{equation*}
    \lVert Nv \rVert_2^2 \le \max_{1 \le i,j \le n}\left\lVert A^{ij}  \right\rVert_2^2 \sum_{i=1}^n \sum_{k,l=1}^n g_{ik} g_{il} \left\lVert v^k \right\rVert_2^2. 
\end{equation*}
It remains the upper bound the RHS in the above inequality. Indeed, we have that
\begin{align*}
    \sum_{i=1}^n \sum_{k,l=1}^n g_{ik} g_{il} \left\lVert v^k \right\rVert_2^2 &= \sum_{i=1}^n \sum_{k=1}^n g_{ik} \left\lVert v^k \right\rVert_2^2 \left(\sum_{l=1}^n  g_{il}\right) \le  \left\lVert G_\infty\right\rVert \sum_{i=1}^n \sum_{k=1}^n g_{ik} \left\lVert v^k \right\rVert_2^2 \\
    &\le \left\lVert G_\infty \right\rVert \sum_{k=1}^n \left(\sum_{i=1}^n g_{ik}\right) \left\lVert v^k \right\rVert_2^2 \le \left\lVert G_\infty \right\rVert \left\lVert G_1 \right\rVert \sum_{k=1}^n \left\lVert v^k \right\rVert_2^2 = \left\lVert G_\infty \right\rVert \left\lVert G_1 \right\rVert.
\end{align*}
Thus \begin{equation*} \sup_{v \colon \left\lVert v\right\rVert =1} \left\lVert
Nv\right\rVert_2^2\le \left\lVert G_\infty \right\rVert \left\lVert G_1 \right\rVert     \max_{i,j}
\left\lVert A^{ij} \right\lVert_2^2,
\end{equation*}
and the conclusion follows.
\end{proof}
\section{Proof of Theorem \ref{thm::main-thm}} \label{sec::main-thm-proof}

In this section we provide the full proof of Theorem \ref{thm::main-thm}. First, we prove the following
result, which will be used to parameterise interactions by the influence bound $\delta_k$. 

\begin{lemma} \label{lem::meloslem}
    In a network game $\game = (\agentset, \edgeset, (u_k, \actionset{k})_{k \in
    \agentset})$, the following holds for any agent $k \in
    \agentset$ action $i \in \actionset{k}$ and strategies $\x_{-k}, \y_{-k} \in \Delta_{-k}$
    \begin{equation*}
        |r_{ki}(\x_{-k}) - r_{ki}(\y_{-k})| \leq \delta_k \sum_{l \in \agentset_k} \lVert \x_l - \y_l \rVert_1.
    \end{equation*}
\end{lemma}
\begin{proof}
    Fix an agent $k$ and define the dummy game $\Tilde{\game}_k = (\agentset_k \cup \{k\},
    (\Tilde{u}_k, \actionset{k})_{k \in \agentset_k \cup \{k\}})$ so that
    $\Tilde{\game}_k$ is composed of only agent $k$ and its neighbours. In addition,
    the rewards are chosen so that $\Tilde{r}_{ki}(\x_{-k}) = r_{ki}(\x_{-k})$
    for all $\x_{-k}$ and  $\Tilde{r}_{li}(\x_{-l}) = \Tilde{r}_{lj}(\x_{-l})$
    for all $l \in \agentset_k$ and all $i, j \in \actionset{k}$. In $\Tilde{\game}_k$,
    the maximum influence bound $\delta := \max_{l \in \agentset_k \cup \{k\}} \delta_l$ is exactly
    $\delta_k$. Then, from \cite{melo:traffic} Proposition 5, the following
    holds in $\Tilde{\game}_k$
    \begin{equation*}
        |r_{ki}(\x_{-k}) - r_{ki}(\y_{-k})| \leq  \sum_{l \neq k} \delta \lVert \x_l - \y_l \rVert_1.
    \end{equation*}
    This translates in the original network game $\game$ to the statement of the Lemma.
\end{proof}
With these results in place, we can prove Theorem \ref{thm::main-thm} in the main paper.
\begin{proof}[Proof of Theorem \ref{thm::main-thm}]
    In order to apply Lemma \ref{lem::ql-conv} we show that, under any of the conditions, the perturbed game $\game^H$ is strongly monotone. To this end, we take
    the derivative of the pseudo-gradient of $\game^H$ which we call the \emph{pseudo-Hessian} given
    by
    \begin{equation*}
        [J(\x)]_{k, l} = D_{\x_l} F_k(\x).
    \end{equation*}
    It follows that, if $\frac{J(\x) + J^\top(\x)}{2}$ is strongly positive definite for all $\x \in
    \Delta$ with any $\alpha > 0$, i.e.  $\x^\top J(\x) \x \geq \alpha$ for all $\x \in \Delta$,
    then $F(\x)$ is strongly monotone with the same constant $\alpha$. We can rewrite the
    pseudo-Hessian as
    \begin{equation*}
        J(\x) = D(\x) + N(\x),
    \end{equation*}
    where $D(\x)$ is a block diagonal matrix with $-D^2_{\x_k \x_k} u_k^H(\x_k, \x_{-k})$ along the
    diagonal. $N(\x)$ is an off-diagonal block matrix with
    \begin{equation*}
        [N(\x)]_{k, l} = \begin{cases}
            - D_{\x_k, \x_l} u_k^H(\x_k, \x_{-k}) &\text{ if } (k, l) \in \edgeset \\
            \zeros &\text{ otherwise}
        \end{cases}.
    \end{equation*}
    In words, $N(x)$ shares the same structure of the adjacency matrix $G$ of the game, except that
    it has $-D_{\x_k, \x_l} u_k^H(\x_k, \x_{-k})$ wherever $G$ takes the value $1$ and the block
    matrix $\zeros$ wherever $G$ has $0$. Next we evaluate these partial differentials. Recall that
    \begin{equation*}
        -u_k^H(\x_k, \x_{-k}) = T_k \langle \x_k, \ln \x_k \rangle - \sum_{(k, l) \in \edgeset} \x_k \cdot A^{kl} \x_l.
    \end{equation*}
    As a result, for all $(k, l) \in \edgeset$, $\left[N(\x)\right]_{k, l} = - A^{kl}$, so that
    $N(\x)$ represents the network interaction. By contrast, $D(\x)$ depends on $T_k$ and is
    independent of the payoffs $u_k$. As such, it measures the strength of the game perturbation.
    Now, let $\Bar{J}(\x)$ be defined as
    \begin{equation*}
        \Bar{J}(\x) = \frac{J(\x) + J^\top(\x)}{2}
                    = D(\x) + \frac{N(\x) + N^\top(\x)}{2}.
    \end{equation*}
    Then, from Proposition \ref{prop::strong-conv} it  follows that $D(\x)$ is strongly positive
    definite with constant $T = \min_{k} T_k$. In particular, this means that $\lambda_{\min} D(\x)
    \geq T$. Finally, applying Weyl's inequality
    \begin{align*}
        \lambda_{\min}(\Bar{J}) &\geq T + \lambda_{\min} \left( \frac{N + N^\top}{2} \right) \\
        &\geq T - \rho\left( \frac{N + N^\top}{2} \right) \\
        &= T - \left\lVert\frac{N + N^\top}{2} \right\rVert_2 \\
        &\geq T - \frac{1}{2} \left\lVert A + B \right\lVert_2  \sqrt{\lVert G \rVert_\infty \lVert G \rVert_1}  \\
        &= T - \frac{1}{2} \left\lVert A + B^\top \right\rVert_2 \left\lVert G \right\rVert_\infty\\
        &= T - \frac{1}{2} \sigma_I \left\lVert G \right\rVert_\infty,
    \end{align*}
    where we employ Propositions \ref{prop::spectral-radius}, Lemma \ref{lem::two-norm} and the fact
    that $G$ is symmetric so that $\lVert G \rVert_\infty = \lVert G \rVert_1$. The matrices $A, B$
    are chosen so that
    \begin{equation*}
        \left\lVert A + B^\top \right\rVert_2 = \max_{(k, l) \in \edgeset} \left\lVert A^{kl} + (A^{lk})^\top \right\rVert_2 = \sigma_I.
    \end{equation*}
    Then, under (\ref{eqn::infty-cond}), $\lambda_{\min}(\Bar{J}(\x)) \geq T - \frac{1}{2} \sigma_I
    \lVert G \rVert_\infty > 0$ and, therefore $F(\x)$ is strongly monotone with constant $T -
    \frac{1}{2} \sigma_I \left\lVert G \right\rVert_\infty$. Using Lemma \ref{lem::ql-conv}, it
    follows that Q-Learning Dynamics converge to a unique QRE.

    To achieve (\ref{eqn::2-cond}) we apply Proposition \ref{prop::kron} which
    yields that
    \begin{align*}
        & T - \left\lVert\frac{N + N^\top}{2} \right\rVert_2 \\
        = &T - \left\lVert\frac{(A + B^\top) \otimes G}{2} \right\rVert_2  \\
        = &T - \frac{1}{2} \left\lVert A + B^\top \right\rVert_2 \left\lVert G \right\rVert_2\\
        = &T - \frac{1}{2} \sigma_I \left\lVert G \right\rVert_2.
    \end{align*}

    Finally, we prove (\ref{eqn::influence-cond}). In this case, it holds that,
    for any $k$ and any $\x, \y \in \Delta$
    \begin{align*}
        (\x_k - \y_k)^\top(F_k(\x) - F_k(\y)) &= (\x_k - \y_k)^\top(T_k \ln \x_k - T_k \ln \y_k) - (\x_k - \y_k)^\top(T_k r_k(\x_{-k}) - T_k r_k(\y_{-k})) \\
        &\geq T_k \lVert \x_k - \y_k \rVert_1^2 - \left| (\x_k - \y_k)^\top \left( r_k(\x_{-k}) - r_k(\y_{-k}) \right) \right| \\
        &\geq T_k \lVert \x_k - \y_k \rVert_1^2 - \lVert \x_k - \y_k \rVert_1 \delta_k \sum_{l \in \agentset_k} \lVert \x_l - \y_l \rVert_1 \\
        &\geq T_k \lVert \x_k - \y_k \rVert_1^2 - \lVert \x_k - \y_k \rVert_1 \delta_k \sum_{l \neq k} [G]_{kl} \lVert \x_l - \y_l \rVert_1 \\
        &= \xi_k (M \xi)_k,
    \end{align*}
    where $\xi = (\x_k - \y_k)_{k \in \agentset}$ and $M = (diag(T_k)_{k \in
    \agentset} - diag(\delta_k)_{k \in \agentset} \cdot G)$. Notice that, to achieve the third inequality, we applied Lemma \ref{lem::meloslem} Then under
    (\ref{eqn::influence-cond}), $M$ is strictly diagonally dominant and so is strictly
    positive definite. Then
    \begin{align*}
        \sum_k (\x_k - \y_k)^\top(F_k(\x) - F_k(\y)) \geq \xi^\top M \xi > 0,
    \end{align*}
    so that $F(\x)$ is strictly monotone. Then, Lemma \ref{lem::ql-conv} can be
    applied to yield convergence of Q-Learning Dynamics.
    \end{proof}

\section{Proofs from Section \ref{sec::e-NE}} \label{sec::e-NE-proof}

\begin{figure*}[t!]
    \captionsetup{justification=centering}
    \centering
    \begin{subfigure}[b]{0.45\textwidth}
        \centering
        \includegraphics[width=0.9\textwidth]{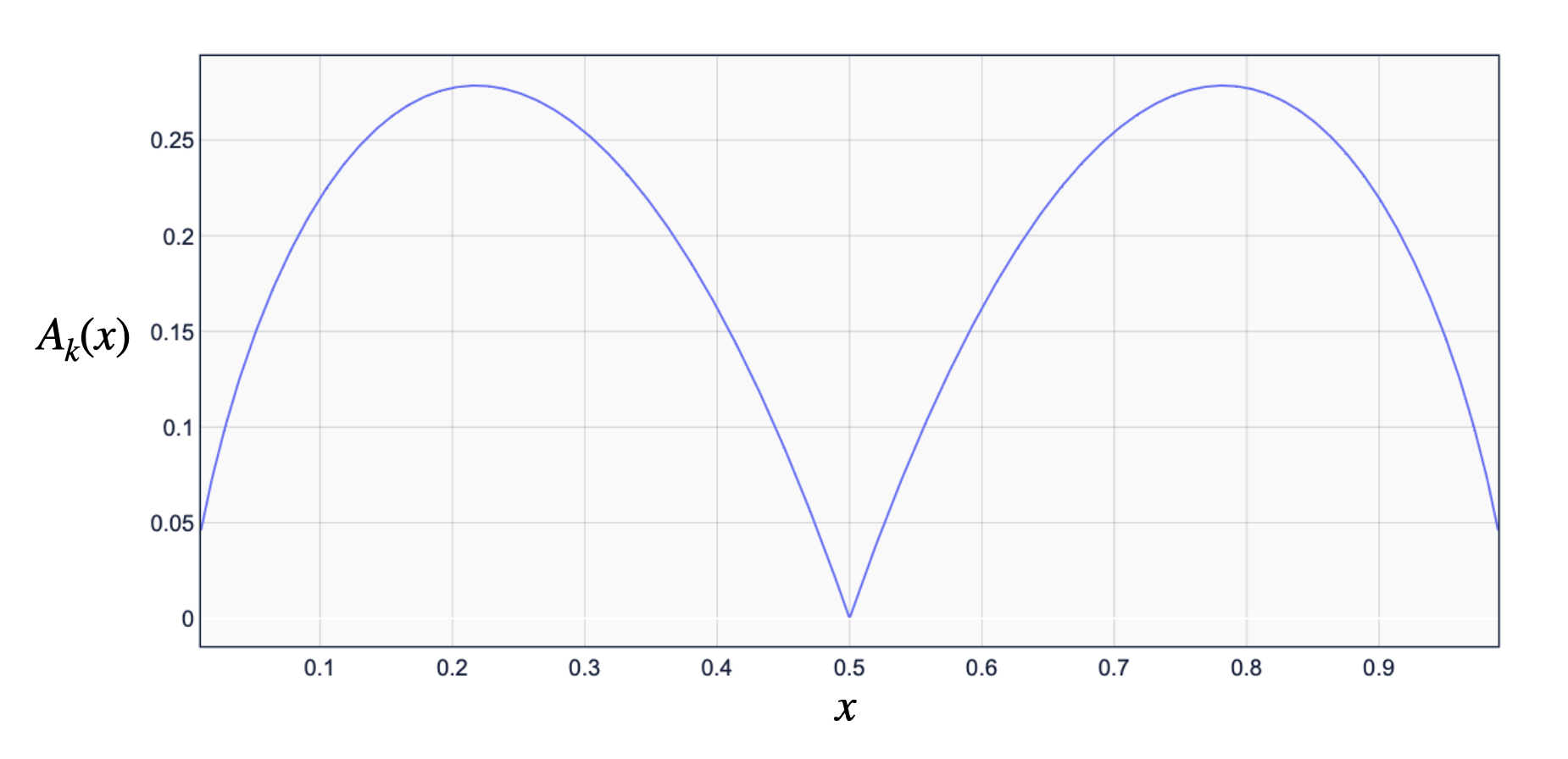}
        \caption*{$n_k = 2$}
    \end{subfigure}
    \begin{subfigure}[b]{0.45\textwidth}
        \centering
        \includegraphics[width=0.6\textwidth]{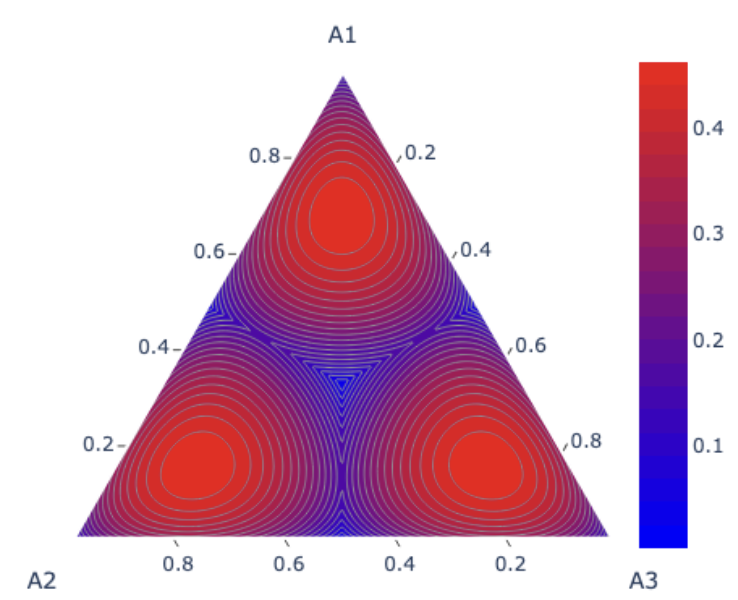}
        \caption*{$n_k = 3$}
    \end{subfigure}

    \caption{$A_k(\x_k)$ plotted on a unit simplex $\Delta_k$ }\label{fig::surprisal_plot}
\end{figure*}

First we show that any QRE $\NE$ is an approximate Nash Equilibrium.

\begin{proof}[Proof of Proposition \ref{prop::e-NE}] 
    We first notice that, for some $\epsilon > 0$, the definition of an $\epsilon$-Nash Equilibrium in Definition \ref{def::e-NE} holds if
    \begin{equation*}
        \max_{k \in \agentset} \max_{i \in \actionset{k}} r_{ki}(\NE_{-k}) - \langle \NE_k, r_k(\NE_{-k}) \rangle = \epsilon.
    \end{equation*}
    Next, recall that a QRE $\NE \in \Delta$ corresponds to an interior fixed point of the Q-Learning Dynamics \cite{piliouras:zerosum}. From this it holds that, for any $k$ and any $i \in \actionset{k}$
    \begin{align*}
        0 &= r_{ki}(\NE_{-k}) - \langle \NE_k, r_k(\NE_{-k}) \rangle + T_k \sum_{j \in \actionset{k}} \bar{x}_{kj} \ln \frac{x_{kj}}{x_{ki}} \\
        r_{ki}(\NE_{-k}) - \langle \NE_k, r_k(\NE_{-k}) \rangle &= -T_k \sum_{j \in \actionset{k}} \bar{x}_{kj} \ln \frac{x_{kj}}{x_{ki}} \\
        &= T_k \ln x_{ki} - \langle \x_k, \ln \x_k \rangle.
    \end{align*}
    As this holds for any $i \in \actionset{k}$, the following holds
    \begin{align*}
        \max_{k \in \agentset} \max_{i \in \actionset{k}} r_{ki}(\NE_{-k}) - \langle \NE_k, r_k(\NE_{-k}) \rangle &= \max_{k \in \agentset} T_k \max_{i \in \actionset{k}}  \ln x_{ki} - \langle \x_k, \ln \x_k \rangle \\
        &= \max_{k \in \agentset} T_k A_k(\NE_k),
    \end{align*}
    so that $\NE$ is an $\epsilon$-Nash Equilibrium with $\epsilon = \max_{k \in \agentset} T_k A_k(\NE_k)$
\end{proof}

    \begin{lemma} [Full version of Lemma \ref{lem::A_k-bound}] \label{lem::A_k-exact}
    Let $\U_k = (1/n_k,\ldots,1/n_k) \in \Delta_k$ and $\e_{ki} \in \Delta_k$ be the canonical basis vector with $i^\text{th}$ entry equal to $1$ and $0$ elsewhere. Then
	\begin{equation}
		\bar{A}_k := \max_{\x_k \in \Delta_k}\left(\max_{i \in \actionset{k}} \ln{x_{ki}} - \sum_{j \in \actionset{k}} x_{kj} \ln{x_{kj}}\right) = \frac{\ln{(n_k - 1)} - \ln\left(W\left(\frac{n_k - 1}{e}\right)\right)}{1 + 1 /W\left(\frac{n_k - 1}{e}\right)} ,\label{eq:ln_minus_entropy}
	\end{equation}
	with equality if $\x^{*} = c \, \e_{ki} + (1-c)\U_k\,$  for any $i \in \actionset{k}$, where $c =1 / \left(W\left(\frac{n_k-1}{e}\right)+1\right)$ and $W(\cdot)$ is the Lambert W function.
    \end{lemma}
    \begin{proof} The $\max$ over $i \in \{1,\ldots,n_k\}$ can be eliminated, as 
        \begin{equation}
            \max_{\x_k \in \Delta_k}\left(\max_{i \in \actionset{k}} \ln{x_{ki}} - \sum_{j \in \actionset{k}} x_{kj} \ln{x_{kj}}\right) = \max_{\x_k \in \Delta_k}\left(\ln{x_{ki}} - \sum_{j \in \actionset{k}} x_{kj} \ln{x_{kj}}\right).\label{eq:no_max_over_i}
        \end{equation} 
    Notice that any value of the LHS realized by $\x_k$ can also be realized by the RHS for $\y_k$, where $\y_k$ is obtained from $\x_k$ by swapping the largest and $i^\text{th}$ entries of $\x_k$. Explicitly, let $i_0 = \arg\max_{j \in \actionset{k}} x_{kj}$ and 
    \begin{equation*}
    y_{kj} = 
        \begin{cases}
            x_{ki_0} &\text{ if } j = i\\
            x_{ki} &\text{ if } j = i_0 \\
            x_{kj} &\text{ otherwise }
        \end{cases}
    \end{equation*}
    Having eliminated the inner maximization, we can formulate the RHS of \eqref{eq:no_max_over_i} as a constrained optimization problem, which we solve in two steps.
    
    Firstly, we show that the $\x^{*}$ maximizing the RHS of \eqref{eq:no_max_over_i} lies on one of the line segments connecting the vertices $\e_{ki}$ of the simplex $\Delta_k$ with the uniform distribution $\U_k$ (see Figure \ref{fig::surprisal_plot}). Next, we determine the position of $\x^{*}$ on this segment, and substitute the value of $\x^{*}$ to obtain the RHS of \eqref{eq:ln_minus_entropy}. We then derive tight bounds on this quantity.
    
    For the first step, consider the Lagrangian 
    $L \colon \Delta_k \times (0,-\infty)  \rightarrow \R \cup \{\infty\}$ given by 
    \begin{equation*}
    L\left(\x_k\right) = \ln{x_{ki}} - \sum_{j \in \actionset{k}} x_{kj} \ln{x_{kj}} + \lambda \left(\sum_{j \in \actionset{k}} x_{kj} -1\right).
    \end{equation*}
    Setting the partial derivatives to $0$, we obtain that
    \begin{align*}
        \frac{\partial L}{\partial x_{ki}} &= -1 + \lambda + 1/x_{kj}  - \ln{x_i} = 0,\\
        \frac{\partial L}{\partial x_{kj}} &= -1 + \lambda  - \ln{x_{kj}} = 0 \text{ for all } j \in \actionset{k} \text{ with } j \neq i, \\
        \frac{\partial L}{\partial \lambda} &= \sum_{j \in \actionset{k}} x_{kj} -1 = 0,
    \end{align*}
    with the solution given by $\x^{*}$ with $x_{kj} = 1/W(e^{1-\lambda})$ and $x_{ki} = e^{-1+\lambda}$ for $j \neq i$, where $W$ is the Lambert W function. Note that $L(\x_k) = -\infty$ for $\x_k$ on the boundary of $\Delta_k$ and further that the mapping $\x_k \to \ln{x_{ki}} - \sum_{j \in \actionset{k}} x_{kj} \ln{x_{kj}}$ is concave, as its Hessian is diagonal with negative entries, hence the stationary point of $L$ gives a maximum. Finally, determining $\lambda$ is equivalent to solving 
    \begin{equation*}
        \sum_{j \in \actionset{k}} x_{kj} = \frac{1}{W(e^{1-\lambda})} + (n_k-1) e^{-1+\lambda} = 1,
    \end{equation*}
    which is intractable to the best of our knowledge, even with modern software such as Mathematica. Nevertheless, we have proved that the maximizer $\x^{*}$ of \eqref{eq:no_max_over_i} lies on the line segment connecting a vertex and the centre of the simplex, which reduces our initial constrained optimization problem to one dimension. Without loss of generality, pick the first vertex and let $\x^{*} = (x,\frac{1-x}{n_k-1},\ldots,\frac{1-x}{n_k-1})$ for some $x \in [0,1]$. We have that
    \begin{equation*}
    \ln{x_i} - \sum_{j \in \actionset{k}} x_{kj} \ln{x_{kj}} = \ln{x} + x \ln{x} + (1-x) \ln{\frac{1-x}{n_k-1}},
    \end{equation*}
    which involves no special functions. By setting the derivative to $0$, we find that this expression is maximized for $x = 1/\left(1+W\left(\frac{n_k-1}{e}\right)\right)$ and that the maximum value is given by
    \begin{equation*}
        \Bar{A}_k = \frac{\ln{(n_k - 1)} - \ln\left(W\left(\frac{n_k - 1}{e}\right)\right)}{1 + 1 /W\left(\frac{n_k - 1}{e}\right)}.
    \end{equation*}
    Finally, we give bounds for $\Bar{A}_k$. \cite{hoorfar:inequalities} proves the sharp bound
    \begin{equation*}
    \ln{x} - \ln{\ln{x}} < W(x) <\ln{x} - \ln{\ln{x}} + \ln{(1+ 1/e)},
    \end{equation*}
    which translates to
    \begin{equation*}\Bar{A}_k < \frac{\ln{n_k} - \ln{\left(\ln{n_k} - \ln{\ln{n_k}}\right)}}{1 + \frac{1}{\ln n_k - \ln{\ln{n_k}}}} < \ln{n_k} - \ln{\left(\ln{n_k} - \ln{\ln{n_k}}\right)} <  \ln n_k.
    \end{equation*}
    \end{proof}

\begin{algorithm}
\caption{Iterative improvement of QRE}\label{alg:cap}
\flushleft \textbf{Input: } Network game $\game = (\agentset, \edgeset, (u_k, \actionset{k})_{k \in \agentset})$; Exploration Rate annealing step $\Delta T$; Maximum number of anneals $M$; Q-Learning horizon $H$; Convergence Window Length $h$; Tolerance $\texttt{tol}$. \\
\flushleft \textbf{Output: } Learned QRE $\NE \in \Delta$ \\
\begin{algorithmic}
\State $T_k \gets \delta_k |\agentset_k| \;$ \algorithmicforall{ $k \in \agentset$} \Comment{or (\ref{eqn::infty-cond}), (\ref{eqn::2-cond})}
\For{$\tau = 1:H$}
    \For{k = 1, \ldots, N}
        \State $Q_{ki} \gets (1 - \alpha_k) Q_{ki} + \alpha_k r_{ki}(\x_{-k})$
        \State $\x_k(\tau) \gets \texttt{softmax}(Q_k / T_k)$
    \EndFor
\EndFor
\State $\NE \gets \x(H)$
\For{t = 1:M} \Comment{or until break statement is reached}
    \For{$k = 1, \ldots, N$}
        \State $\epsilon_k \gets T_k A_k(\NE_k) $ \Comment{from (\ref{eqn::A_k})}
    \EndFor
\State $l = \arg\max_{k \in \agentset} \epsilon_k$ \Comment{ties broken arbitrarily}
\State $T_l \gets T_l - \Delta T$
\For{$\tau = 1:H$}
    \For{k = 1, \ldots, N}
        \State $Q_{ki} \gets (1 - \alpha_k) Q_{ki} + \alpha_k r_{ki}(\x_{-k})$
        \State $\x_k(\tau) \gets \texttt{softmax}(Q_k / T_k)$
    \EndFor
\EndFor
\State $V \gets \max_{k, i} \left\{\frac{\max_{\tau \in H} x_{ki}(\tau) - \min_{\tau \in H} x_{ki}(\tau)}{\min_{\tau \in H} x_{ki}(\tau)} \right\}$
\If{$V < \texttt{tol}$}
    \State $\NE \gets \x(H)$
\Else
    \State \textbf{break} 
\EndIf
\EndFor
\end{algorithmic}
\end{algorithm}

\section{Additional Experiments} \label{sec::add-expt}

In this section, we present additional experiments on the behaviour of Q-Learning in Network Games, as well as on the exploration update scheme. In Figure \ref{fig::mismatching-traj}, we examine a Network Mismatching Game, which was analysed in \cite{kleinberg:nashbarrier} as an example of limit cycle behaviour in replicator dynamics. Here, the payoff to each agent $k$ is given as
\begin{align*}
    &u_k(\x_k, \x_{-k}) = \x_k^\top \A \x_l, \; l = k-1 \mod N, \\
    &A = \begin{pmatrix}
        0 & 1 \\ M & 0
    \end{pmatrix}, \; M \geq 1
\end{align*}
From Figure \ref{fig::mismatching-traj} it is clear that, as exploration rates increase, the dynamics are driven towards a QRE from all initial conditions.

Next, we present additional experiments on the exploration updating scheme in Section \ref{sec::iterative-scheme}. In particular, we apply the scheme to a Network Mismatching Game with 5 agents. We plot the exploitability (\ref{eqn::exploitability}) and $\epsilon$ (\ref{eqn::e-NE}) over $150,000$ iterations of learning. In both cases it is again clear that the distance to Nash decreases as the exploration updating scheme is applied. In the case that $M=2$, the scheme is applied until (\ref{eqn::conv-criteria}) fails at approx. $60,000$ iterations, whilst in the case $M=4$, agents learn for $80,000$ iterations before the dynamics are considered unstable. In Figure \ref{fig::centralised-scheme-traj} we plot the trajectories of Q-Learning using the first action played by three representative agents. The dynamics move between QRE as the exploration rates are adjusted, however stability of the dynamic is maintained.

\begin{figure*}[t]
	\centering
 	\begin{subfigure}[b]{0.225\textwidth}
		\centering
		\includegraphics[width=\textwidth]{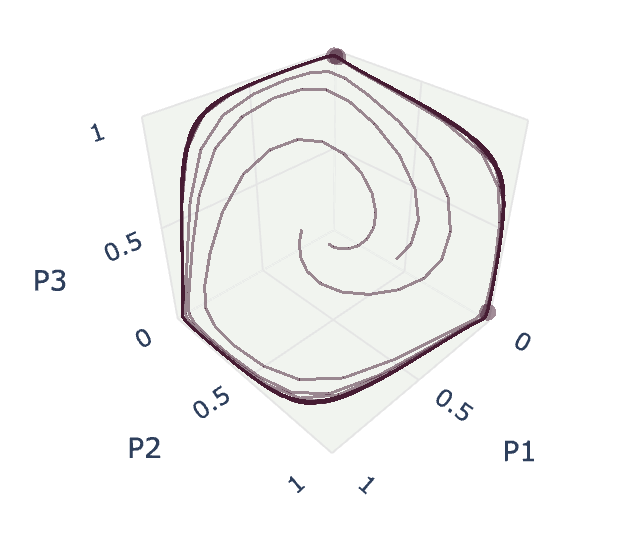}
		\caption*{$T = 0.15$}
	\end{subfigure}
	\begin{subfigure}[b]{0.225\textwidth}
		\centering
		\includegraphics[width=\textwidth]{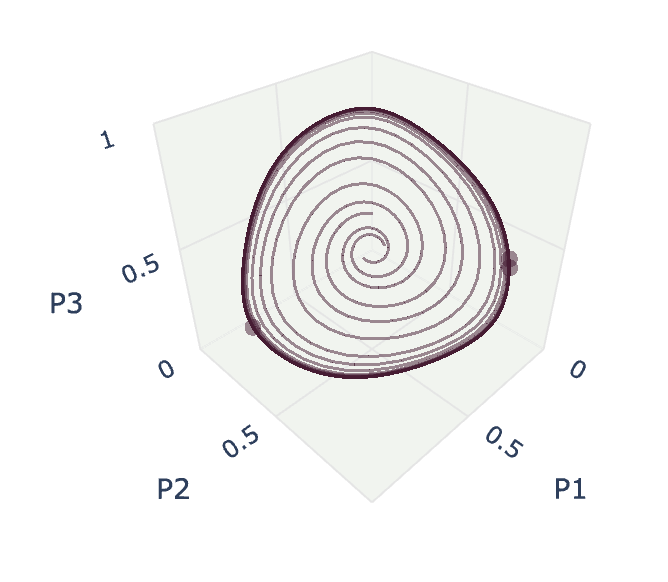}
		\caption*{$T = 0.3$}
	\end{subfigure}
	\begin{subfigure}[b]{0.225\textwidth}
		\centering
		\includegraphics[width=\textwidth]{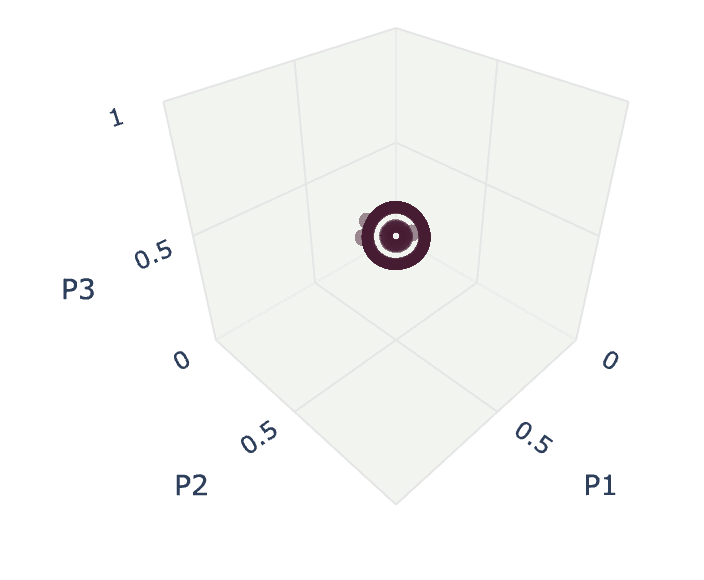}
		\caption*{$T = 0.4$}
	\end{subfigure}
	\begin{subfigure}[b]{0.225\textwidth}
		\centering
		\includegraphics[width=\textwidth]{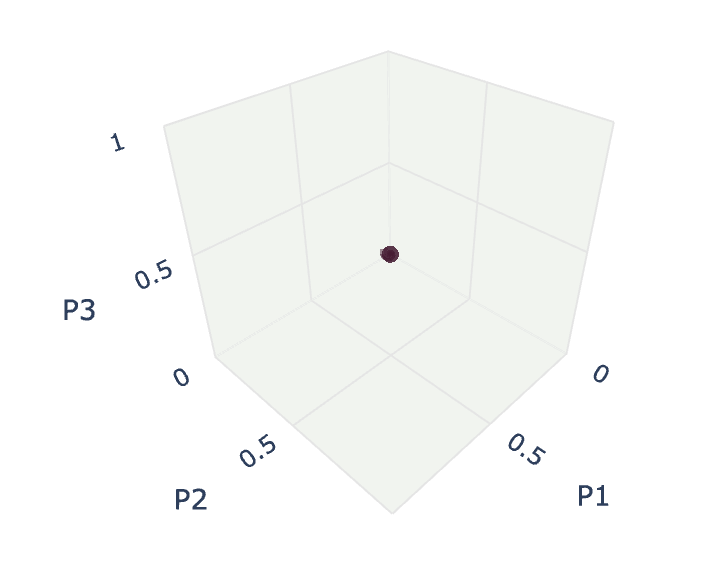}
		\caption*{$T = 0.55$}
	\end{subfigure}
 \caption{Trajectories of Q-Learning in a three agent Network Mismatching Game with $M = 2$. Axes denote the probabilities with which each player chooses their first action.}\label{fig::mismatching-traj}
\end{figure*}

\begin{figure*}[t!]
	\centering
	\begin{subfigure}[b]{0.45\textwidth}
		\centering
		\includegraphics[width=\textwidth]{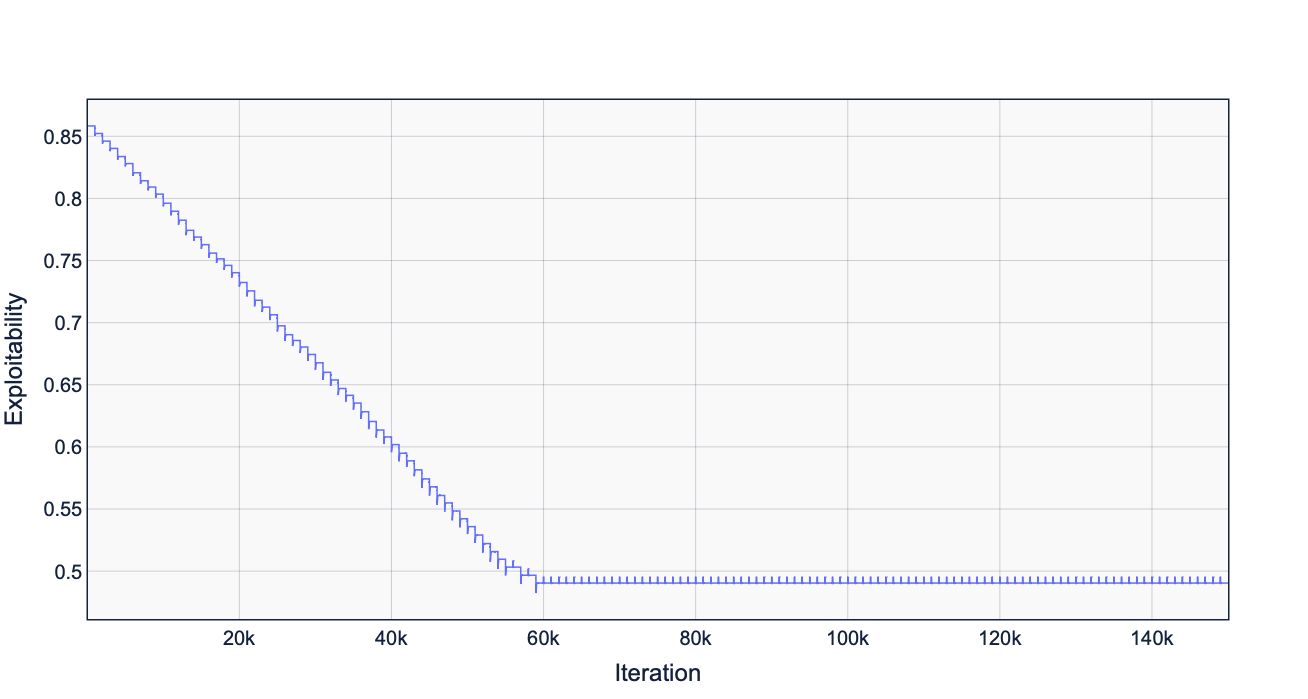}
	\end{subfigure}
	\begin{subfigure}[b]{0.45\textwidth}
		\centering
		\includegraphics[width=\textwidth]{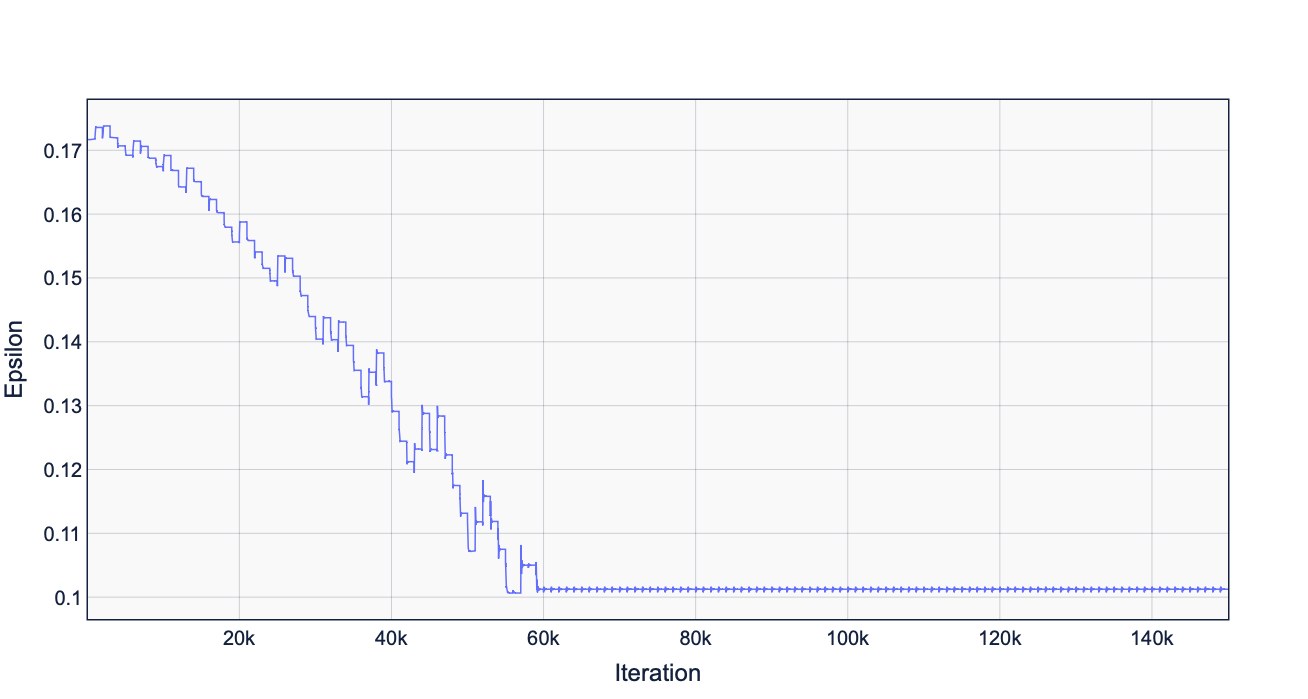}
	\end{subfigure}

  	\begin{subfigure}[b]{0.45\textwidth}
		\centering
		\includegraphics[width=\textwidth]{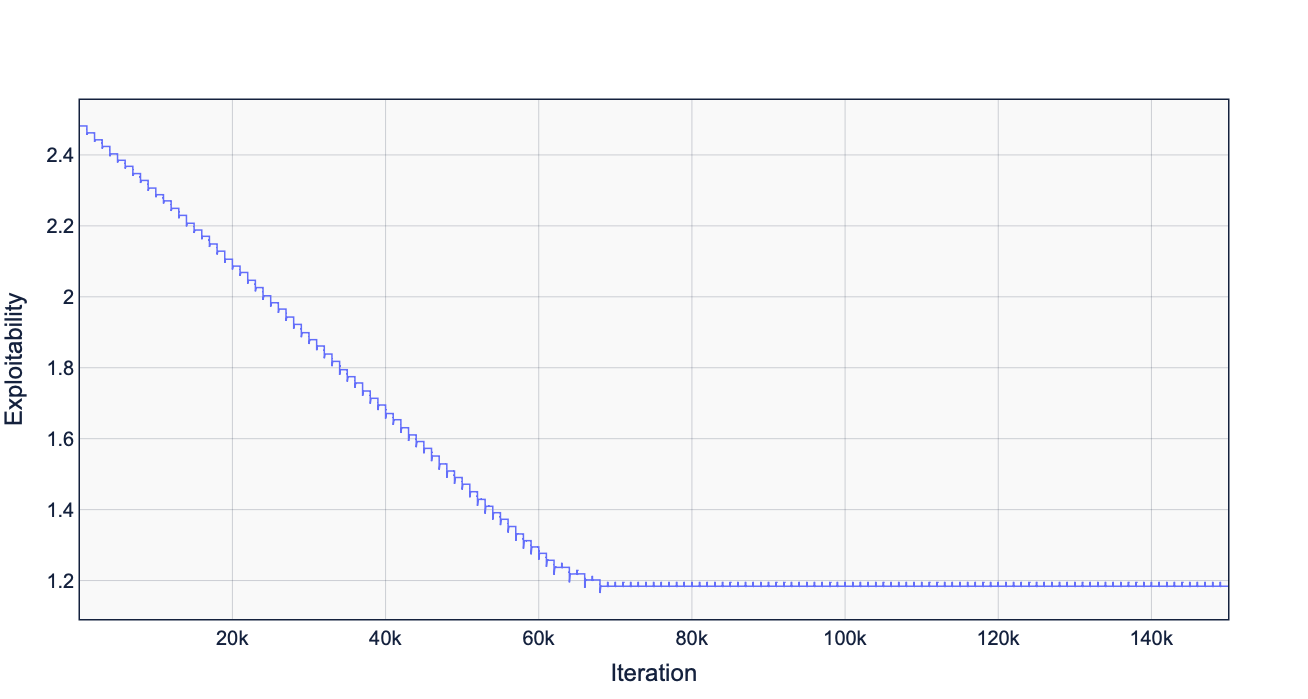}
	\end{subfigure}
	\begin{subfigure}[b]{0.45\textwidth}
		\centering
		\includegraphics[width=\textwidth]{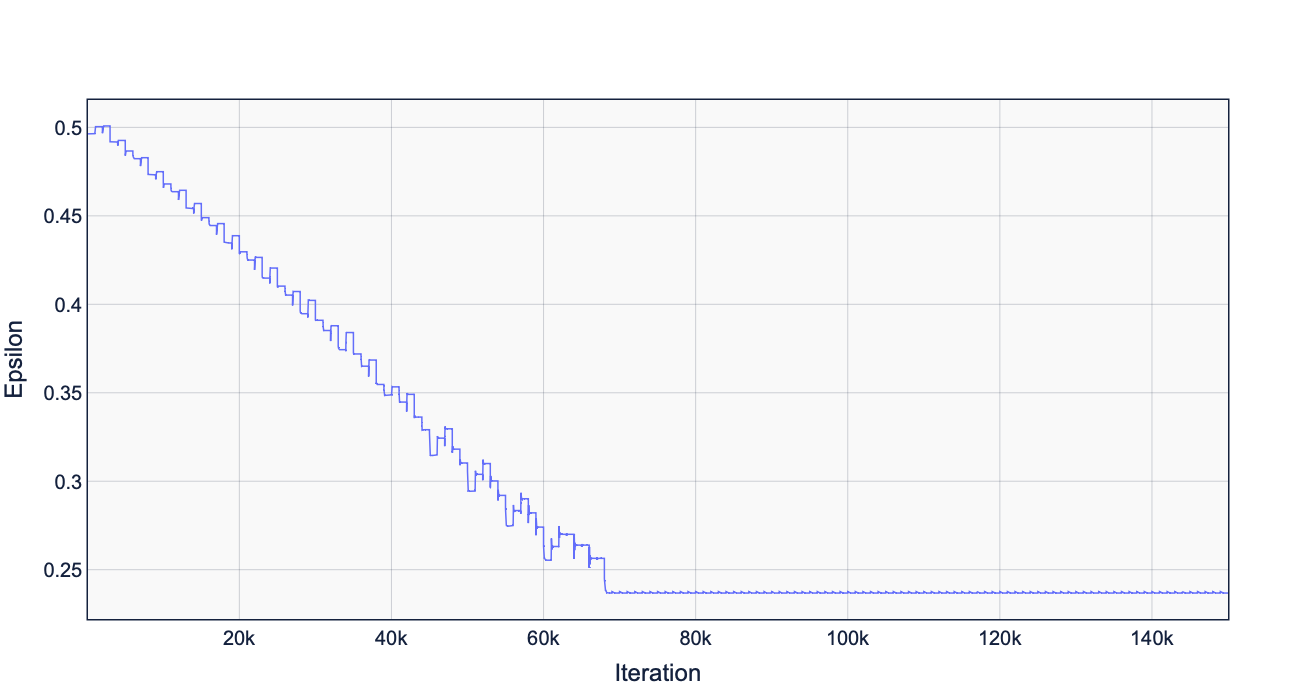}
	\end{subfigure}
 
 \caption{Measures of `closeness' to Nash Equilibrium as the exploration update scheme is applied to the Network Mismatching Game with five agents and (Top) $M=2$ (Bottom) $M=4$. (Left) Distance to NE measured by exploitability (\ref{eqn::exploitability}) of the joint strategy $\x(t)$. (Right) $\epsilon$ as defined by (\ref{eqn::e-NE}). Both metrics decreases as exploration rates are updated until condition (\ref{eqn::conv-criteria}) fails, after which learning is halted.}
 \label{fig::centralised-scheme-mismatch}
\end{figure*}

\begin{figure}[t!]
	\centering
	\begin{subfigure}[b]{0.45\columnwidth }
		\centering
		\includegraphics[width=\columnwidth]{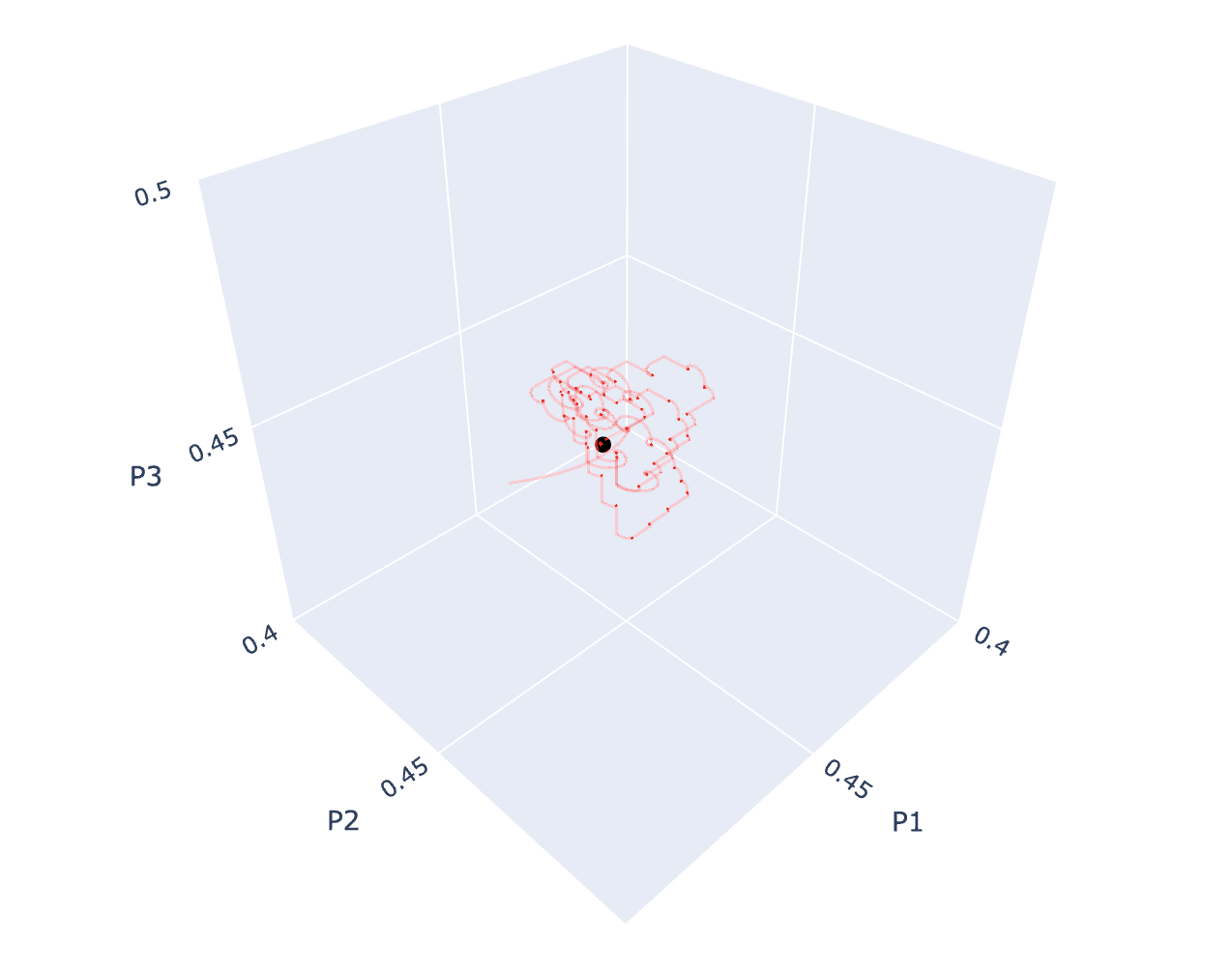}
	\end{subfigure}
	\begin{subfigure}[b]{0.45\columnwidth}
		\centering
		\includegraphics[width=\columnwidth]{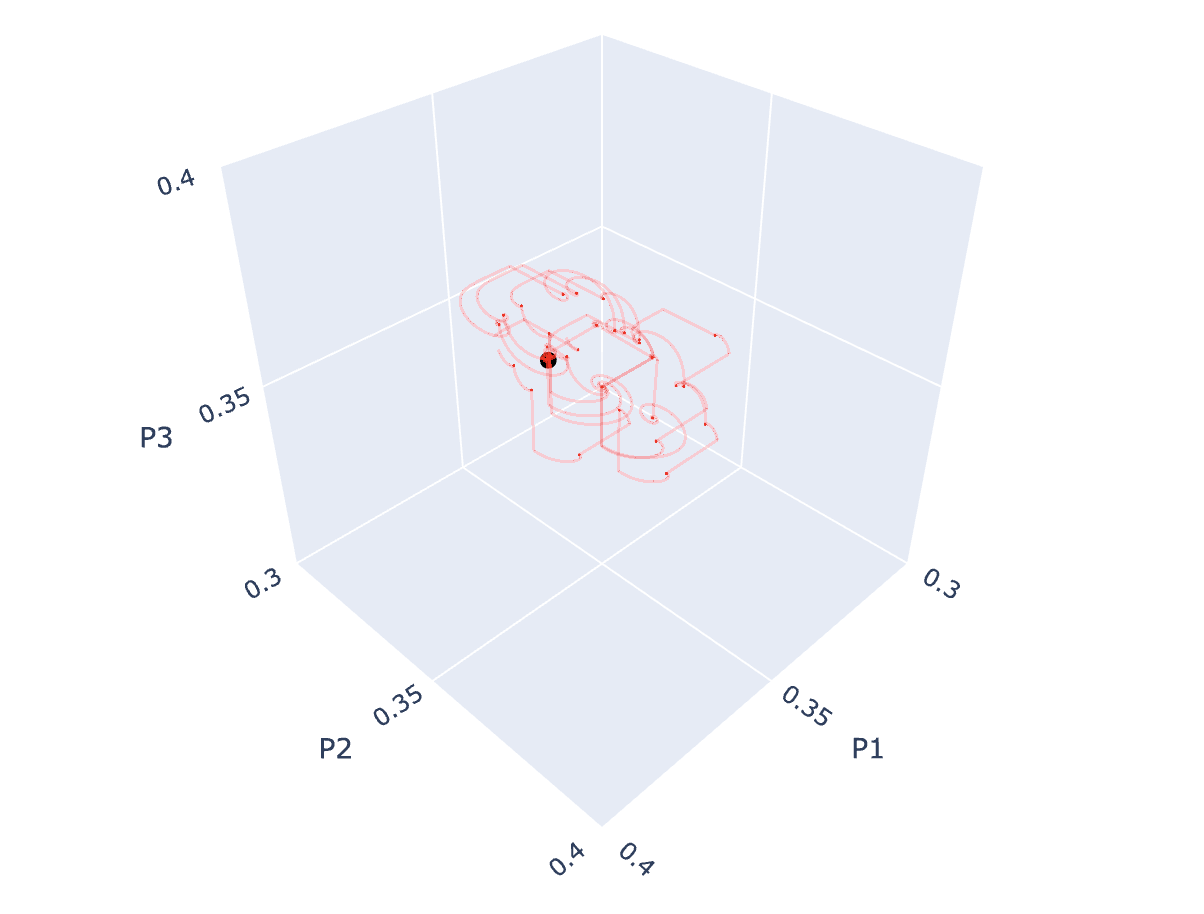}
	\end{subfigure}
 
 \caption{Trajectories of Q-Learning generated as the centralised scheme is applied to (Left) Mismatching Game with $M=2$ (Right) Mismatching Game with $M=4$.}
 \label{fig::centralised-scheme-traj}
\end{figure}

\end{document}